\documentclass[a4paper]{siamltex}

\usepackage{stmaryrd}
\usepackage{amsmath}
\usepackage{amssymb}
\usepackage{colonequals}

\usepackage{url}
\usepackage{hyperref}
\usepackage{graphicx}
\usepackage[numbers]{natbib}

\newcommand{\norm}[1]{\left\| {#1} \right\|}
\newcommand{\pd}{\partial}
\newcommand{\dif}{\, {\rm d}}
\newcommand{\vect}[1]{\boldsymbol{#1}}
\newcommand{\brac}[1]{\left( #1 \right)}

\title{Energy stable and momentum conserving hybrid finite element method
for the incompressible Navier--Stokes equations}
\author{Robert Jan Labeur\thanks{Faculty of Civil Engineering and Geosciences,
                                 Delft University of Technology,
                                 Stevinweg 1,
                                 2628 CN Delft,
                                 The Netherlands
                                 ({\tt r.j.labeur@tudelft.nl})}
\and Garth N. Wells\thanks{Department of Engineering,
                           University of Cambridge,
                           Trumpington Street,
                           Cambridge CB2 1PZ,
                           United Kingdom
                           ({\tt gnw20@cam.ac.uk})}}

\begin{document}
\maketitle
\begin{abstract}
A hybrid method for the incompressible Navier--Stokes equations is
presented. The method inherits the attractive stabilizing mechanism of
upwinded discontinuous Galerkin methods when momentum advection becomes
significant, equal-order interpolations can be used for the velocity
and pressure fields, and mass can be conserved locally.  Using
continuous Lagrange multiplier spaces to enforce flux continuity across
cell facets, the number of global degrees of freedom is the same as for a
continuous Galerkin method on the same mesh.  Different from our earlier
investigations on the approach for the Navier--Stokes equations, the
pressure field in this work is discontinuous across cell boundaries. It
is shown that this leads to very good local mass conservation and, for
an appropriate choice of finite element spaces, momentum conservation.
Also, a new form of the momentum transport terms for the method is
constructed such that global energy stability is guaranteed, even in the
absence of a point-wise solenoidal velocity field.  Mass conservation,
momentum conservation and global energy stability are proved for the
time-continuous case, and for a fully discrete scheme.  The presented
analysis results are supported by a range of numerical simulations.
\end{abstract}



\section{Introduction}

A method that combines attractive features of discontinuous and continuous
Galerkin finite element methods for the incompressible Navier Stokes
equations was presented in \citet{Labeur:2007} and further extended to
the case of moving domains and free-surface flows in \citet{Labeur:2009}.
The method incorporated naturally the evaluation of upwinded advective
fluxes on cell facets, in the same spirit as discontinuous Galerkin
methods, thereby stabilizing flows with significant momentum advection,
and it is possible to use equal-order polynomial bases for the velocity
and pressure components.  However, the number of global degrees of
freedom on a given mesh is the same as for a continuous Galerkin method
using the same polynomial orders.  The issue regarding the significantly
greater number of global degrees of freedom for low- to moderate-order
discontinuous Galerkin methods compared to continuous Galerkin methods
is thus circumvented. However, the method in \citep{Labeur:2007} was
restricted to continuous pressure fields, and it could not be proven that
the method is globally energy stable.  These short-comings are addressed
in this paper, with a formulation presented that permits discontinuous
pressure fields, is globally energy stable, conserves momentum and
has excellent local mass conservation properties.

The key to the methodology that we present for constructing schemes
is the postulation of cell-wise balances, subject to weakly enforced
boundary conditions.  The boundary condition to be satisfied (weakly)
is provided by a function that lives on cell facets only.
An equation for this extra field is furnished
by insisting on weak continuity of the associated `numerical' flux.
The concept of weak enforcement of flux continuity across cell facets
is central in hybridized finite element methods (for an overview see
\citep{Cockburn:2009c}).  A feature of these methods is that functions
on cells are linked to functions on neighboring cells \emph{only} via
functions that live on cell facets, and not directly via the flux terms.
Therefore, functions on cells can be eliminated locally in favor of
the functions that live on cell facets only (via static condensation),
thus reducing the number of globally coupled degrees of freedom.  If the
functions enforcing the continuity of the fluxes, and which live only
on cell facts, are discontinuous, then point-wise continuous fluxes
can be obtained for suitably chosen function spaces.  In contrast, in
our method we advocate the use of facet functions that are continuous,
which leads to a significant reduction in the number of globally coupled
degrees of freedom, since the local elimination procedure will lead to
a global problem of the same size as a corresponding continuous Galerkin
method.  Yet, it is straightforward to demonstrate local momentum and
mass conservation, in terms of the numerical fluxes, as is typical of
discontinuous Galerkin methods.  Also the stabilizing mechanism of the
flux formulation, involving the advection terms and the pressure-velocity
coupling, are directly inherited from the discontinuous Galerkin method
and lead to favorable stability properties.

We are not alone in considering methods that draw on both discontinuous
and continuous Galerkin methods. \citet{Hughes:2006} developed a
method for the advection--diffusion equation, and the formulation in
\citep{Labeur:2007} for the advection-diffusion equation reduces to
that of \citet{Hughes:2006} in the advective limit. In the diffusive
case, there is a subtle difference, with the diffusive flux in
\citet{Hughes:2006} being upwinded, whereas a centered approach is
used in \citet{Labeur:2007}.  Simulations using the approach for
the advection--diffusion equation exhibited very good stability
properties, minimal dissipation and standard convergence rates.
For the case of the linear advection--diffusion-reaction equation,
stability (via an inf--sup condition) and convergence at a rate
of $k+1$ in the diffusive limit and $k+1/2$ in the advective limit
was later proved~\citep{Wells:2010}.  In the context of hybridized
methods, Cockburn and co-workers have published a number of works
(e.g.~\citep{cockburn:2009b,nguyen:2010}) that share features with the
methodology that we consider. A hybrid field on cell interfaces is
presented in~\citet{Egger:2010} for the advection--diffusion problem.
\citet{Guzey:2007} present a hybrid continuous-discontinuous finite
element method for elliptic problems, coined embedded discontinuous
Galerkin method, which is conceptually related to the method
in \citet{Labeur:2007}.

The method formulated and analyzed in this work is an extension of the
method presented in \citet{Labeur:2007} for the advection--diffusion
equation and for the incompressible Navier--Stokes equations.  Unlike in
our previous efforts \citep{Labeur:2007,Labeur:2009}, we consider here
pressure fields that are discontinuous across cell facets.  The impact of
this on the local (cell-wise) mass conservation properties of the method
is demonstrated.  Another feature that distinguishes the formulation
developed in this work from our earlier work for the Navier--Stokes
equations is the use of a skew-symmetric form of the advective term. The
derivation of the skew-symmetric formulation is not trivial in the
considered setting, but it is shown that it brings the advantage of
guaranteeing stability in terms of the total kinetic energy, even when
the velocity field is not point-wise solenoidal. The combination of
discontinuous pressure fields and skew-symmetric advection terms leads to
a method that for equal-order basis functions preserves mass and momentum
and is also stable in terms of the total kinetic energy.  The analysis
results that we present are supported by a number of computations for both
the Stokes and incompressible Navier--Stokes equations.  The computer
code necessary to reproduce all examples presented in this work is
available in the supporting material \citep{supporting_material} under
a GNU public license.

The remainder of this work is structured as follows. We first define
concretely the problem of interest, and then develop a semi-discrete
finite element formulation. Some properties of the semi-discrete problem
are then analyzed. This is followed by a particular fully-discrete
formulation, and it is shown that the considered properties of the
semi-discrete problem are inherited by the fully discrete problem. This
is followed by numerical examples, after which conclusions are drawn.
\section{Incompressible Navier--Stokes equations}
\label{sec:method}

We consider a domain of interest $\Omega \subset \mathbb{R}^{d}$, where
$1 \le d \le 3$ is the spatial dimension.  The boundary $\pd \Omega$ is
assumed sufficiently smooth and the outward unit normal vector on $\pd
\Omega$ is denoted by $\vect{n}$. The boundary is partitioned such that
$\Gamma_{D} \cap \Gamma_{N} = \pd\Omega$ and $\Gamma_{D} \cup \Gamma_{N} =
\emptyset$. The time interval of interest is $I = \left(0, t_{N} \right]$.

The non-dimensional incompressible Navier--Stokes problem on
$\Omega \times I$ reads: given the viscosity $\nu$, the forcing term
$\vect{f} \colon \Omega \times I \rightarrow \mathbb{R}^{d}$, the momentum
flux $\vect{h} \colon \Gamma_{N} \times I \rightarrow \mathbb{R}^{d}$ and
the solenoidal initial condition $\vect{u}_{0} \colon \Omega \rightarrow
\mathbb{R}^{d}$, find the velocity field $\vect{u} \colon \Omega \times
I \rightarrow \mathbb{R}^{d}$ and the pressure field $p \colon \Omega
\times I \rightarrow \mathbb{R}$ such that
\begin{align}
  & \frac{\pd \vect{u}}{\pd t} + \nabla \cdot \vect{\sigma} = \vect{f}
      \quad \text{on} \ \Omega \times I,
   \label{eqn:strong_momentum_eqn}
\\
  & \vect{\sigma} = p \vect{I} - 2 \nu \nabla^{s} \vect{u} + \vect{u} \otimes \vect{u}
   \quad \text{on} \ \Omega \times I,
\\
  & \nabla \cdot \vect{u} = 0 \quad \text{on} \ \Omega \times I,
   \label{eqn:inc_constraint}
\\
  & \vect{u} = \vect{0} \quad \text{on} \ \Gamma_{D} \times I,
\\
  & \vect{\sigma} \vect{n} - \max \brac{\vect{u} \cdot \vect{n}, 0} \vect{u}
  = \vect{h} \quad \text{on} \ \Gamma_{N} \times I,
   \label{eqn:NeumannBC}
\\
  & \vect{u}\brac{\vect{x}, 0}=\vect{u}_{0}\brac{\vect{x}} \quad \text{on} \ \Omega,
\end{align}
where $\vect{\sigma}$ is the momentum flux, $\vect{I}$ is the identity tensor,
$\nabla^{s} \vect{u} = \brac{\nabla \vect{u} + \nabla \vect{u}^{T}}/2$
is the symmetric gradient, in which
$[\nabla \vect{u}]_{ij} = \pd u_i /\pd x_j$, and
$[\vect{u} \otimes \vect{u}]_{ij} = u_{i}u_{j}$ .
The Neumann boundary condition has been formulated such that on portions
of $\Gamma_{N}$ on which $\vect{u} \cdot \vect{n} < 0$ (inflow boundaries)
the total momentum flux is prescribed, while on portions of $\Gamma_{N}$
on which $\vect{u} \cdot \vect{n} \ge 0$ (outflow boundaries) only the
diffusive part of the momentum flux is prescribed.
\section{Finite element method}

The hybrid finite element method is defined in this section. The essence
of the method is posing all balance equations cell-wise in a weak sense,
with a suitably constructed numerical flux, and complementing the
cell-wise balance laws by a global equation enforcing weak continuity
of the numerical flux across cell facets.

\subsection{Definitions}
\label{sec:definitions}

We consider a triangulation $\mathcal{T}$ of the domain $\Omega$ into
open, non-overlapping sub-domains $K$ (cells).  The outward unit normal
vector on the boundary $\pd K$ of each cell is denoted by $\vect{n}$.
Adjacent cells share a common facet $F$, and $\mathcal{F}= \bigcup F$
is the union of all facets, including the exterior boundary facets.
A measure of the size of a cell $K$ is denoted by $h_{K}$. When evaluated
on a shared facet, $h_{K}$ is used to imply the average cell size measure
of the adjacent cells.

Consider first the vector finite element spaces $V_h$ and $\Bar{V}_h$:
\begin{align}
   & V_{h} \colonequals \left \lbrace \vect{v}_h \in \left[ L^{2} \brac{\mathcal{T}} \right]^{d},
   \vect{v}_h \in \left[ P_{k} \brac{K} \right]^{d} \ \forall \ K \in \mathcal{T}\right\rbrace,
   \label{eqn:def-V}
\\
   & \Bar{V}_{h}
   \colonequals
   \left\lbrace \Bar{\vect{v}}_h \in \left[ L^{2} \brac{\mathcal{F}} \right]^{d},
   \; \Bar{\vect{v}}_h \in \left[ P_{\Bar{k}} \brac{F} \right]^{d} \ \forall \ F \in \mathcal{F},
   \; \Bar{\vect{v}}_h = \vect{0} \ \text{on} \ \Gamma_{D} \right\rbrace,
   \label{eqn:def-Bar-V}
\end{align}
where $P_{k}(K)$ denotes the space of Lagrange polynomials on $K$
of order $k > 0$, and $P_{\Bar{k}}(F)$ denotes the space of
Lagrange polynomials on $F$ of order $\Bar{k} \geq 0$.  The space $V_{h}$
contains vector-valued functions that are discontinuous across cell
boundaries, while functions in $\Bar{V}_h$
are defined on cell facets only.  Furthermore, functions in
$\Bar{V}_h$ satisfy the homogeneous Dirichlet boundary condition
on~$\Gamma_D$.  Scalar finite element spaces $Q_h$ and~$\Bar{Q}_h$
are defined by:
\begin{align}
  &  Q_{h} \colonequals \left \lbrace q_h \in L^{2} \brac{\mathcal{T}},
     q_h \in P_{m} \brac{K} \ \forall \ K \in \mathcal{T} \right\rbrace,
\\
   & \Bar{Q}_{h} \colonequals \left \lbrace \Bar{q}_h \in L^{2}\brac{\mathcal{F}},
     \Bar{q}_h \in P_{\Bar{m}} \brac{F} \ \forall \ F \in \mathcal{F} \right\rbrace,
   \label{eqn:def-Bar-Q}
\end{align}
where the polynomial orders $m \geq 0$ and $\Bar{m} \geq 0$.  Mirroring
the vector spaces, $Q_h$ contains functions that are discontinuous across
cell facets, while functions in $\Bar{Q}_h$ are defined on cell facets only.

For algorithmic reasons, it may be advantageous to compute with the finite
element spaces
\begin{align}
   & \Bar{V}^{\star}_{h}
   \colonequals \Bar{V}_{h} \cap \left[ H^{1}(\mathcal{F}) \right]^{d},
\\
   & \Bar{Q}^{\star}_{h} \colonequals \Bar{Q}_{h} \cap H^{1}(\mathcal{F}),
\end{align}
in place of $\Bar{V}_{h}$ and $\Bar{Q}_{h}$, respectively,
using polynomial orders $\Bar{k} \geq 1$ and $\Bar{m} \geq 1$.
This will be discussed in
Section~\ref{sec:algorithmic-aspects}, and all computational
results presented in Section~\ref{sec:examples} will employ facet
functions that are continuous.
\subsection{Semi-discrete weak local/global balances}

We formulate now a semi-discrete finite element problem by considering
what we will refer to as `local' and `global' equations.  The `local'
equations solve the problem cell-wise in which the velocity and pressure
boundary conditions are provided by auxiliary fields that live on cell
facets only. To determine the fields that live on cell facets, `global'
equations are formulated by requiring weak continuity of the mass and
momentum fluxes across element interfaces. The methodology behind the
construction of the formulation is elucidated by presenting a collection
of Galerkin problems for the various balances, after which the considered
Galerkin finite element problem is completely and formally defined.

\subsubsection{Local/global continuity equation}
\label{sec:cont_eq}

A Galerkin approximation of the incompressibility
constraint~\eqref{eqn:inc_constraint} in a cell-wise fashion requires
that the approximate velocity $\vect{u}_{h} \in V_h$ satisfies
\begin{equation}
  \sum_{K}\int_{K} \vect{u}_{h} \cdot \nabla q_{h} \dif x
 -  \sum_{K} \int_{\pd K} \Hat{\vect{u}}_{h} \cdot \vect{n} \; q_{h} \dif s
 = 0
 \quad \forall \ q_{h} \in Q_{h},
\label{eqn:weak_incomp}
\end{equation}
where $\Hat{\vect{u}}_{h}$ is the `numerical' mass flux on $\pd K$,
and is chosen to be
\begin{equation}
  \Hat{\vect{u}}_{h} =  \vect{u}_{h} - \frac{\beta h_{K}}{\nu + 1} \brac{\Bar{p}_{h} - p_{h}} \vect{n},
\label{eqn:def_uhat}
\end{equation}
in which $p_{h} \in Q_{h}$ and $\Bar{p}_{h} \in \Bar{Q}_{h}$ are pressure
fields and $\beta > 0$ is a parameter required for stability when using
equal-order basis functions for the velocity components and pressure
fields. When using a lower polynomial order for the pressure
field relative to the velocity field it is possible to use $\beta =
0$~\citep{hansbo:2002}.  Penalization of the pressure jump was used by
\citet{hughes:1987} to stabilize equal-order methods with discontinuous
pressure for the Stokes equation, and by other authors for discontinuous
Galerkin methods \citep{cockburn:2002,cockburn:2009}. However,
different from \citet{hughes:1987}, we add a non-dimensional unit
viscosity to the term in the denominator to permit consideration of
the inviscid limit.  Central in~\eqref{eqn:def_uhat} is that the
pressure-stabilizing term involves the difference between $p_{h}$
and $\Bar{p}_{h}$, rather than the jump in $p_{h}$ across a facet
as in other works \citep{hughes:1987,cockburn:2002,cockburn:2009}.
Equation~\eqref{eqn:weak_incomp} is local in the sense that there is no
direct interaction between $p_{h}$ on neighboring cells. This is a key
feature of the method with practical implications that will be elaborated
upon in Section~\ref{sec:algorithmic-aspects}.

The numerical mass flux in~\eqref{eqn:def_uhat} is not unique on cell
facets; it can take on different values on different sides of a facet.
This is in contrast with standard discontinuous Galerkin methods, in
which the numerical mass flux is constructed such that it is uniquely
defined on facets.  A `global' continuity equation is now furnished
by insisting that the numerical mass flux $\Hat{\vect{u}}_h$
be weakly continuous across cell facets, in that it satisfies
\begin{equation}
  \sum_{K}\int_{\pd K} \Hat{\vect{u}}_{h} \cdot \vect{n}  \; \Bar{q}_{h} \dif s
 -  \int_{\pd \Omega} \Bar{\vect{u}}_{h} \cdot \vect{n}  \; \Bar{q}_{h} \dif s
      = 0 \quad \forall \ \Bar{q}_{h} \in \Bar{Q}_{h},
\label{eqn:global_incomp}
\end{equation}
where $\Bar{\vect{u}}_h \in \Bar{V}_h$.  Note that
\eqref{eqn:global_incomp} implies that $\Hat{\vect{u}}_{h} \cdot \vect{n}
= \Bar{\vect{u}}_{h} \cdot \vect{n}$ (weakly) on $\pd \Omega$.

\subsubsection{Local/global momentum balance in conservative form}
\label{sec:cons_mom_eq}
At time $t$ and given the forcing term $\vect{f} \in \left[ L^{2}
\brac{\mathcal{T}} \right]^{d}$, the viscosity $\nu$, the velocity
$\Bar{\vect{u}}_{h} \in \Bar{V}_{h}$, and pressures $p_{h} \in Q_{h}$
and $\Bar{p}_{h} \in \Bar{Q}_{h}$, consider a Galerkin approximation
$\vect{u}_{h} \in V_{h}$ that satisfies the following weak formulation
of the momentum balance~\eqref{eqn:strong_momentum_eqn}:
\begin{multline}
   \int_{\Omega} \frac{\pd \vect{u}_{h}}{\pd t} \cdot \vect{v}_{h} \dif x
  - \sum_{K}\int_{K} \vect{\sigma}_{h} \colon \nabla \vect{v}_{h} \dif x
  + \sum_{K} \int_{\pd K} \Hat{\vect{\sigma}}_{h} \vect{n} \cdot \vect{v}_{h} \dif s
\\
  + \sum_{K} \int_{\pd K} 2\nu \brac{\Bar{\vect{u}}_{h} - \vect{u}_{h}}
    \cdot \nabla^{s} \vect{v}_{h}  \, \vect{n} \dif s
   = \int_{\Omega} \vect{f} \cdot \vect{v}_{h} \dif x
   \quad \forall \ \vect{v}_{h} \in V_{h},
\label{eqn:local_momentum_eqn_conservative}
\end{multline}
where the momentum flux $\vect{\sigma}_{h}$ on cells is given by
\begin{equation}
 \vect{\sigma}_{h} = \vect{\sigma}\brac{\vect{u}_{h}, p_{h}}
   = p_{h} \vect{I} - 2 \nu \nabla^{s} \vect{u}_{h}
   + \vect{u}_{h} \otimes \vect{u}_{h},
 \label{eqn:def_mom_flux}
\end{equation}
and the `numerical' momentum flux $\Hat{\vect{\sigma}}_{h}$ on cell boundaries
is given by
\begin{equation}
   \Hat{\vect{\sigma}}_{h} = \Hat{\vect{\sigma}}_{a, h} + \Hat{\vect{\sigma}}_{d, h},
\label{eqn:int_flux}
\end{equation}
where the advective flux $\Hat{\vect{\sigma}}_{a, h}$ is
\begin{equation}
   \Hat{\vect{\sigma}}_{a, h}
   = \Hat{\vect{\sigma}}_{a}\brac{\vect{u}_{h}, \Bar{\vect{u}}_{h}, p_{h}, \Bar{p}_{h}}
    = \vect{u}_{h} \otimes \Hat{\vect{u}}_{h}
     + \brac{\Bar{\vect{u}}_{h} - \vect{u}_{h}} \otimes \lambda \Hat{\vect{u}}_{h},
\label{eqn:adv_int_flux}
\end{equation}
in which $\Hat{\vect{u}}_{h}$ is given by~\eqref{eqn:def_uhat},
$\lambda$ is a function that takes on a value of either one or zero
and is defined below, and the diffusive flux $\Hat{\vect{\sigma}}_{d, h}$ is
\begin{equation}
  \Hat{\vect{\sigma}}_{d, h}
  = \Hat{\vect{\sigma}}_{d}\brac{\vect{u}_{h}, \Bar{\vect{u}}_{h}, \Bar{p}_{h}}
       = \Bar{p}_{h} \vect{I} - 2\nu \nabla^{s} \vect{u}_{h}
                    - \frac{\alpha}{h_{K}} 2 \nu
             \brac{\Bar{\vect{u}}_{h} - \vect{u}_{h} } \otimes \vect{n}.
  \label{eqn:diff_int_flux}
\end{equation}
The function $\lambda$ takes on a value of one on inflow cell
boundaries (where $\Hat{\vect{u}} \cdot \vect{n} < 0$), and takes
on a value of zero on outflow cell boundaries (where $\Hat{\vect{u}} \cdot
\vect{n} \geq 0$).  The formulation for the advective interface flux
involves upwinding since $\Hat{\vect{\sigma}}_{a}=\vect{u} \otimes
\Hat{\vect{u}}$ on outflow cell boundaries and $\Hat{\vect{\sigma}}_{a}
= \Bar{\vect{u}} \otimes \Hat{\vect{u}}$ on inflow cell boundaries.
In \eqref{eqn:local_momentum_eqn_conservative}, the fourth term
on the left-hand side ensures symmetry of the diffusion operator
(see~\citet{Arnold:2002}).  In~\eqref{eqn:diff_int_flux},
$\alpha$ is a penalty parameter, and such a term is typical
of interior penalty methods.  Just as for standard interior
penalty methods, the role of the penalty term in this context
is to ensure stability, as detailed in~\citet{Wells:2010}.
Equation~\eqref{eqn:local_momentum_eqn_conservative} is `local' in the
sense that the weak momentum balance equation is posed cell-wise.

As with the numerical mass flux~\eqref{eqn:def_uhat}, the numerical
momentum flux~\eqref{eqn:int_flux} is not single-valued on cell facets. A
`global' momentum balance equation is therefore formulated by insisting on
continuity of the numerical flux $\Hat{\vect{\sigma}}_{h}$ across element
facets.  This continuity constraint is imposed weakly by requiring that,
for a given flux boundary condition
$\vect{h} \in \left[ L^2\brac{\Gamma_N} \right]^d$,
\begin{multline}
  \sum_{K} \int_{\pd K} \Hat{\vect{\sigma}}_{h} \vect{n} \cdot \Bar{\vect{v}}_{h} \dif s
  =
  \int_{\Gamma_{N}} \brac{1 - \lambda}
  \brac{\Bar{\vect{u}}_{h} \otimes \Bar{\vect{u}}_{h}} \vect{n} \cdot \Bar{\vect{v}}_{h} \dif s
\\
+ \int_{\Gamma_{N}} \vect{h} \cdot \Bar{\vect{v}}_{h}  \dif s
  \quad \forall \ \Bar{\vect{v}}_{h} \in \Bar{V}_{h}.
\label{eqn:flux_continuity_cons}
\end{multline}
The above equation implies that the numerical momentum flux
$\Hat{\vect{\sigma}}_{h} \vect{n}$ and the momentum flux on $\Gamma_{N}$,
given by $\vect{h} + \brac{1 - \lambda} \brac{\Bar{\vect{u}}_{h} \otimes
\Bar{\vect{u}}_{h}} \vect{n}$, coincide in a weak sense.

\subsubsection{Local/global momentum balance in advective form}
\label{sec:adv_mom_eq}

We now rephrase the conservative forms of the local and global momentum
balance equations into advective formats with a view to formulating a
skew-symmetric version of the advective terms.

Considering first the local momentum equation,
substitution of the fluxes~\eqref{eqn:def_mom_flux},
\eqref{eqn:int_flux} and~\eqref{eqn:adv_int_flux}
into~\eqref{eqn:local_momentum_eqn_conservative} yields
\begin{multline}
   \int_{\Omega} \frac{\pd \vect{u}_{h}}{\pd t} \cdot \vect{v}_{h} \dif x
  - \sum_{K}\int_{K} \vect{\sigma}_{d, h} \colon \nabla \vect{v}_{h} \dif x
  - \sum_{K}\int_{K} \brac{\vect{u}_{h} \otimes \vect{u}_{h}} \colon \nabla \vect{v}_{h} \dif x
\\
  + \sum_{K} \int_{\pd K} \Hat{\vect{\sigma}}_{d, h} \vect{n} \cdot \vect{v}_{h} \dif s
  + \sum_{K} \int_{\pd K} \brac{\Hat{\vect{u}}_{h} \cdot \vect{n}} \,
     \vect{u}_{h} \cdot \vect{v}_{h} \dif s
\\
  + \sum_{K} \int_{\pd K} \lambda \brac{\Hat{\vect{u}}_{h} \cdot \vect{n}} \,
             \brac{\Bar{\vect{u}}_{h} - \vect{u}_{h}} \cdot \vect{v}_{h} \dif s
\\
  + \sum_{K} \int_{\pd K} 2\nu \brac{\Bar{\vect{u}}_{h} - \vect{u}_{h}}
    \cdot \nabla^{s} \vect{v}_{h}  \, \vect{n} \dif s
  = \int_{\Omega} \vect{f} \cdot \vect{v}_{h} \dif x,
\label{eqn:cons_mom_eq}
\end{multline}
in which $\vect{\sigma}_{d, h} = p_{h}\vect{I} - 2\nu \nabla^s
\vect{u}_{h}$ is the diffusive flux on cells. Applying partial integration
to the advective terms on $K$ in~\eqref{eqn:cons_mom_eq},
\begin{multline}
   \int_{\Omega} \frac{\pd \vect{u}_{h}}{\pd t} \cdot \vect{v}_{h} \dif x
  - \sum_{K}\int_{K} \vect{\sigma}_{d, h} \colon \nabla \vect{v}_{h} \dif x
  + \sum_{K}\int_{K} \brac{\nabla \vect{u}_{h} \, \vect{u}_{h}} \cdot \vect{v}_{h} \dif x
\\
  + \sum_{K}\int_{K} \brac{\nabla \cdot \vect{u}_{h}} \vect{u}_{h} \cdot \vect{v}_{h} \dif x
  + \sum_{K} \int_{\pd K} \brac{\brac{\Hat{\vect{u}}_{h} - \vect{u}_{h}}
    \cdot \vect{n}} \, \vect{u}_{h} \cdot \vect{v}_{h} \dif s
\\
  + \sum_{K} \int_{\pd K} \Hat{\vect{\sigma}}_{d, h} \vect{n} \cdot \vect{v}_{h} \dif s
  + \sum_{K} \int_{\pd K} \lambda \brac{\Hat{\vect{u}}_{h} \cdot \vect{n}} \,
             \brac{\Bar{\vect{u}}_{h} - \vect{u}_{h}} \cdot \vect{v}_{h} \dif s
\\
  + \sum_{K} \int_{\pd K} 2\nu \brac{\Bar{\vect{u}}_{h} - \vect{u}_{h}}
    \cdot \nabla^{s} \vect{v}_{h}  \, \vect{n} \dif s
  = \int_{\Omega} \vect{f} \cdot \vect{v}_{h} \dif x.
\label{eqn:cons_mom_eq_2}
\end{multline}
We choose to discard the integrals involving $\nabla \cdot \vect{u}_{h}$
and $\brac{\Hat{\vect{u}}_{h} - \vect{u}_{h}} \cdot \vect{n}$, which
by virtue of the continuity equation~\eqref{eqn:weak_incomp} and under
an appropriate regularity assumption on the exact solution will not
disturb consistency of a Galerkin scheme (this will be addressed
formally in Section~\ref{sec:properties}).  A reduced version
of~\eqref{eqn:cons_mom_eq_2} now reads:
\begin{multline}
   \int_{\Omega} \frac{\pd \vect{u}_{h}}{\pd t} \cdot \vect{v}_{h} \dif x
  + \sum_{K}\int_{K} \brac{\nabla \vect{u}_{h} \, \vect{u}_{h}} \cdot \vect{v}_{h} \dif x
  - \sum_{K}\int_{K} \vect{\sigma}_{d, h} \colon \nabla \vect{v}_{h} \dif x
\\
  + \sum_{K} \int_{\pd K} \lambda \brac{\Hat{\vect{u}}_{h} \cdot \vect{n}} \,
             \brac{\Bar{\vect{u}}_{h} - \vect{u}_{h}} \cdot \vect{v}_{h} \dif s
+ \sum_{K} \int_{\pd K} \Hat{\vect{\sigma}}_{d, h} \vect{n} \cdot \vect{v}_{h} \dif s
\\
  + \sum_{K} \int_{\pd K} 2\nu \brac{\Bar{\vect{u}}_{h} - \vect{u}_{h}}
    \cdot \nabla^{s} \vect{v}_{h}  \, \vect{n} \dif s
  = \int_{\Omega} \vect{f} \cdot \vect{v}_{h} \dif x.
 \label{eqn:momentum_advective_form}
\end{multline}

Considering next the global momentum balance,
inserting the expressions for the numerical flux~\eqref{eqn:int_flux}
and~\eqref{eqn:adv_int_flux} into the global flux
continuity equation~\eqref{eqn:flux_continuity_cons} yields,
\begin{multline}
  \sum_{K} \int_{\pd K} \Hat{\vect{\sigma}}_{d, h} \vect{n} \cdot \Bar{\vect{v}}_{h} \dif s
 + \sum_{K} \int_{\pd K}  \brac{\Hat{\vect{u}}_{h} \cdot \vect{n}}
               \vect{u}_{h} \cdot \Bar{\vect{v}}_{h} \dif s
\\
+ \sum_{K} \int_{\pd K}  \lambda \brac{\Hat{\vect{u}}_{h} \cdot \vect{n}}
    \brac{\Bar{\vect{u}}_{h} - \vect{u}_{h}} \cdot \Bar{\vect{v}}_{h} \dif s
\\
-  \int_{\Gamma_{N}} \brac{1 - \lambda} \brac{\Bar{\vect{u}}_{h} \cdot \vect{n}}
    \Bar{\vect{u}}_{h} \cdot \Bar{\vect{v}}_{h} \dif s
= \int_{\Gamma_{N}} \vect{h} \cdot \Bar{\vect{v}}_{h} \dif s.
\label{eqn:global_cons_eq}
\end{multline}
The second integral in~\eqref{eqn:global_cons_eq} can be expanded as
\begin{multline}
  \sum_{K} \int_{\pd K}  \brac{\Hat{\vect{u}}_{h} \cdot \vect{n}}
               \vect{u}_{h} \cdot \Bar{\vect{v}}_{h} \dif s
= \sum_{K} \int_{\pd K}  \brac{\Hat{\vect{u}}_{h} \cdot \vect{n}}
               \brac{\vect{u}_{h} - \Bar{\vect{u}}_{h}} \cdot \Bar{\vect{v}}_{h} \dif s \\
 +  \sum_{K} \int_{\pd K}  \brac{\brac{\Hat{\vect{u}}_{h} - \Bar{\vect{u}}_{h}} \cdot \vect{n} }
               \Bar{\vect{u}}_{h} \cdot \Bar{\vect{v}}_{h} \dif s
  + \int_{\pd\Omega} \brac{\Bar{\vect{u}}_{h} \cdot \vect{n}}
    \Bar{\vect{u}}_{h} \cdot \Bar{\vect{v}}_{h} \dif s,
 \label{eqn:global_intermediate_result}
\end{multline}
where we have used that $\Bar{\vect{u}}_h$ is single-valued on cell
facets, by definition. Discarding the term involving
$\brac{\Hat{\vect{u}}_{h} - \Bar{\vect{u}}_{h}} \cdot \vect{n}$, which is
consistent with continuity equation~\eqref{eqn:global_incomp},
and substituting~\eqref{eqn:global_intermediate_result}
into~\eqref{eqn:global_cons_eq} leads to the following advective
form of the global momentum equation:
\begin{multline}
  \sum_{K} \int_{\pd K} \Hat{\vect{\sigma}}_{d, h} \vect{n} \cdot \Bar{\vect{v}}_{h} \dif s
 - \sum_{K} \int_{\pd K} \brac{1 - \lambda} \brac{\Hat{\vect{u}}_{h} \cdot \vect{n}}
               \brac{\Bar{\vect{u}}_{h} - \vect{u}_{h}} \cdot \Bar{\vect{v}}_{h} \dif s
\\
+  \int_{\Gamma_{N}} \lambda \brac{\Bar{\vect{u}}_{h} \cdot \vect{n}}
    \Bar{\vect{u}}_{h} \cdot \Bar{\vect{v}}_{h} \dif s
= \int_{\Gamma_{N}} \vect{h} \cdot \Bar{\vect{v}}_{h} \dif s,
\label{eqn:global_adv_momentum_eq}
\end{multline}
where we have taken into account that $\Bar{\vect{v}}_{h} = \vect{0}$
on~$\pd\Omega \backslash \Gamma_{N}$.

\subsection{Semi-discrete finite element formulation}
\label{sec:combined_method}

We define now a collection of functionals that together will define
a complete finite element problem. For convenience, the  notation
$\vect{U} \colonequals \brac{\vect{u}, \Bar{\vect{u}}, p, \Bar{p}}$
will be used.

Based on the local continuity equation~\eqref{eqn:weak_incomp}, we define
the functional:
\begin{equation}
  F_{\rm c}(\vect{U}; q)
    \colonequals
  \sum_{K}\int_{K} \vect{u} \cdot \nabla q \dif x
  -  \sum_{K} \int_{\pd K} \Hat{\vect{u}} \cdot \vect{n} \; q \dif s,
\label{eqn:local_weak_continuity_L}
\end{equation}
where $F_{\rm c}$ is linear in $q$. Similarly, from the global continuity
equation~\eqref{eqn:global_incomp}, the functional
\begin{equation}
  \Bar{F}_{\rm c}(\vect{U}; \Bar{q})
  \colonequals
  \sum_{K}\int_{\pd K} \Hat{\vect{u}} \cdot \vect{n} \; \Bar{q} \dif s
 -  \int_{\pd \Omega} \Bar{\vect{u}} \cdot \vect{n} \; \Bar{q} \dif s,
\label{eqn:global_weak_continuity_L}
\end{equation}
is defined, where $\Bar{F}_{\rm c}$ is linear in $\Bar{q}$.
For the momentum equations, we define a local momentum
balance functional that is a weighted combination of the local
conservative balance~\eqref{eqn:cons_mom_eq} and the local advective
balance~\eqref{eqn:momentum_advective_form}:
\begin{multline}
 F_{\rm m}(\vect{U}; \vect{v})
  \colonequals
    \int_{\Omega} \frac{\pd \vect{u}}{\pd t} \cdot \vect{v} \dif x
  - \chi \sum_{K}\int_{K} \brac{\vect{u} \otimes \vect{u}} \colon \nabla \vect{v} \dif x
\\
  + \brac{1-\chi}  \sum_{K}\int_{K} \brac{\nabla \vect{u} \, \vect{u}} \cdot \vect{v} \dif x
  + \chi \sum_{K} \int_{\pd K} \brac{\Hat{\vect{u}} \cdot \vect{n}} \, \vect{u} \cdot \vect{v} \dif s
\\
  + \sum_{K} \int_{\pd K} \lambda \brac{\Hat{\vect{u}} \cdot \vect{n}} \,
             \brac{\Bar{\vect{u}} - \vect{u}} \cdot \vect{v} \dif s
  - \sum_{K}\int_{K} \vect{\sigma}_{d} \colon \nabla \vect{v} \dif x
\\
  + \sum_{K} \int_{\pd K} \Hat{\vect{\sigma}}_{d} \vect{n} \cdot \vect{v} \dif s
  + \sum_{K} \int_{\pd K} 2\nu \brac{\Bar{\vect{u}}-\vect{u}}
    \cdot \nabla^{s} \vect{v}  \, \vect{n} \dif s
  - \int_{\Omega} \vect{f} \cdot \vect{v} \dif x,
\label{eqn:local_weak_momentum_L}
\end{multline}
where $\chi \in \left[0,1\right]$ and $F_{\rm m}$ is
linear in~$\vect{v}$.  In the same fashion, the global
momentum flux continuity equations~\eqref{eqn:global_cons_eq}
and~\eqref{eqn:global_adv_momentum_eq} are weighted and summed,
leading to,
\begin{multline}
 \Bar{F}_{\rm m}(\vect{U}; \Bar{\vect{v}})
    \colonequals
      \chi \sum_{K} \int_{\pd K}  \brac{\Hat{\vect{u}} \cdot \vect{n}}
               \vect{u} \cdot \Bar{\vect{v}} \dif s
- \brac{1-\chi} \sum_{K} \int_{\pd K}  \brac{\Hat{\vect{u}} \cdot \vect{n}}
               \brac{\Bar{\vect{u}}-\vect{u}} \cdot \Bar{\vect{v}} \dif s
\\
+ \sum_{K} \int_{\pd K}  \lambda \brac{\Hat{\vect{u}} \cdot \vect{n}}
    \brac{\Bar{\vect{u}}-\vect{u}} \cdot \Bar{\vect{v}} \dif s
+  \sum_{K} \int_{\pd K} \Hat{\vect{\sigma}}_{d} \vect{n} \cdot \Bar{\vect{v}} \dif s
\\
-  \int_{\Gamma_{N}} \brac{\chi - \lambda} \brac{\Bar{\vect{u}} \cdot \vect{n}}
    \Bar{\vect{u}} \cdot \Bar{\vect{v}} \dif s
- \int_{\Gamma_{N}} \vect{h} \cdot \Bar{\vect{v}}  \dif s,
\label{eqn:global_weak_momentum_L}
\end{multline}
where $\Bar{F}_{\rm m}$ is linear in $\Bar{\vect{v}}$.

Defining now
\begin{equation}
F(\vect{U}; \vect{W})
 \colonequals
    F_{\rm m}(\vect{U}; \vect{v})
  + \Bar{F}_{\rm m}(\vect{U}; \Bar{\vect{v}})
  + F_{\rm c}(\vect{U}; q)
  + \Bar{F}_{\rm c}(\vect{U}; \Bar{q}),
\end{equation}
where $\vect{W} = \brac{\vect{v}, \Bar{\vect{v}}, q, \Bar{q}}$, a
semi-discrete finite element problem at time $t$ involves:
given the forcing term
$\vect{f} \in \left[L^2 \brac{\Omega} \right]^d$
the boundary condition
$\vect{h} \in \left[L^2 \brac{\Gamma_N} \right]^d$
and the viscosity $\nu$, find $\vect{U}_{h} \in V_{h}
\times \Bar{V}_{h} \times Q_{h} \times \Bar{Q}_{h}$ such that
\begin{equation}
  F(\vect{U}_{h}; \vect{W}_{h}) = 0
  \quad \forall \ \vect{W}_{h} \in V_{h} \times \Bar{V}_{h}
        \times Q_{h} \times \Bar{Q}_{h}.
\label{eqn:fe_problem}
\end{equation}
This completes the formulation of the semi-discrete finite element
problem.
\section{Properties of the semi-discrete formulation}
\label{sec:properties}

We now demonstrate the consistency, mass conservation, momentum
conservation and energy stability properties of the method for the
semi-discrete formulation in~\eqref{eqn:fe_problem}.  The presented
results hold for the spaces $\Bar{V}_{h}$ and $\Bar{Q}_{h}$,
and deliberately also for the more restrictive case in which
$\Bar{V}_{h}$ and $\Bar{Q}_{h}$ are replaced by $\Bar{V}^{\star}_{h}$
and $\Bar{Q}^{\star}_{h}$, respectively, which we advocate in
practice and will use in numerical examples.

\begin{proposition}[consistency]
\label{prop:consistency}
If at a given time $t$, $\vect{u} \in \brac{H^{2}\brac{\Omega}}^{d}$
and $p \in H^{1}\brac{\Omega}$ solve
equations~\eqref{eqn:strong_momentum_eqn}--\eqref{eqn:NeumannBC},
and $\Bar{\vect{u}} = \gamma(\vect{u})$ and $\Bar{p} = \gamma(p)$ on
$\mathcal{F}$, where $\gamma$ is a trace operator,
then
\begin{equation}
  F(\vect{U}; \vect{W}_{h}) = 0
      \quad \forall \ \vect{W}_{h}
          \in V_{h} \times \Bar{V}_{h} \times Q_{h} \times \Bar{Q}_{h},
\label{eqn:consistency}
\end{equation}
for any $\chi \in [0, 1]$.
\end{proposition}
\begin{proof}
Considering first $\vect{v}_{h} = \vect{0}$, $\Bar{\vect{v}}_{h} = \vect{0}$
and $\Bar{q}_{h} = 0$, applying integration by parts to \eqref{eqn:consistency}
leads to
\begin{equation}
  \sum_{K}\int_{K}  (\nabla \cdot \vect{u}) \, q_{h} \dif x
 -  \sum_{K} \int_{\pd K} \frac{\beta h_{K}}{1+\nu} \brac{\Bar{p}-p} \, q_{h} \dif s
  = 0 \quad \forall \ q_{h} \in Q_{h},
\label{eqn:consistency_proof_c_local}
\end{equation}
which holds due to satisfaction of \eqref{eqn:inc_constraint} and
$\Bar{p} = \gamma(p)$.  Setting $\vect{v}_{h} = \vect{0}$,
$\Bar{\vect{v}}_{h} = \vect{0}$ and $q_{h} = 0$ in \eqref{eqn:consistency},
\begin{equation}
  \sum_{K}\int_{\pd K \backslash \pd \Omega} \Hat{\vect{u}} \cdot \vect{n}  \, \Bar{q}_{h} \dif s
 +  \int_{\pd \Omega} \Hat{\vect{u}} \cdot \vect{n}  \, \Bar{q}_{h} \dif s
 -  \int_{\pd \Omega} \Bar{\vect{u}} \cdot \vect{n}  \, \Bar{q}_{h} \dif s
 = 0 \quad \forall \ \Bar{q}_{h} \in \Bar{Q}_{h},
\label{eqn:consistency_proof_c_global}
\end{equation}
which holds due to the regularity of $\vect{u}$ and because
$\Hat{\vect{u}} = \gamma (\vect{u})$.

Setting $\Bar{\vect{v}}_{h} = \vect{0}$, $q_{h} = 0$ and $\Bar{q}_{h}
= 0$ in~\eqref{eqn:consistency} and applying integration by parts,
\begin{multline}
   \int_{\Omega} \brac{ \frac{\pd \vect{u}}{\pd t}
   + \nabla \cdot \vect{\sigma} - \vect{f}} \cdot \vect{v}_{h} \dif x
   - \brac{1-\chi} \int_{\Omega} \brac{\nabla \cdot \vect{u}} \vect{u} \cdot \vect{v}_{h} \dif x
\\
   + \chi \sum_K \int_{\pd K} \brac{\brac{\Hat{\vect{u}} - \vect{u}} \cdot \vect{n}}
     \vect{u}\cdot\vect{v}_h \dif s
  + \sum_{K} \int_{\pd K} \lambda \brac{\Hat{\vect{u}} \cdot \vect{n}}
             \brac{\Bar{\vect{u}} - \vect{u}} \cdot \vect{v}_{h} \dif s
\\
  + \sum_K \int_{\pd K} \brac{\Bar{p}-p} \vect{n} \cdot \vect{v}_h \dif s
  - \sum_{K} \int_{\pd K} \frac{\alpha}{h_K} 2 \nu
     \brac{\Bar{\vect{u}} -\vect{u}} \otimes \vect{n} \cdot \vect{v}_{h} \dif s
\\
  + \sum_{K} \int_{\pd K} 2\nu \brac{\Bar{\vect{u}} - \vect{u}} \cdot \nabla^{s} \vect{v}_{h} \vect{n} \dif s
  = 0 \quad \forall \ \vect{v}_{h} \in V_{h},
\label{eqn:consistency_proof_m_local}
\end{multline}
which holds due to $\vect{u}$ and $p$ satisfying
equations~\eqref{eqn:strong_momentum_eqn} and~\eqref{eqn:inc_constraint},
the regularity of~$\vect{u}$ and because $\Bar{p}=\gamma(p)$ and
$\Bar{\vect{u}} = \Hat{\vect{u}} = \gamma(\vect{u})$.  Finally, setting
$\vect{v}_{h} = \vect{0}$, $q_{h} = 0$ and $\Bar{q}_{h} = 0$
in~\eqref{eqn:consistency},
\begin{multline}
 \sum_{K} \int_{\pd K}  \vect{\sigma} \vect{n} \cdot \Bar{\vect{v}}_{h} \dif s
+ \sum_{K} \int_{\pd K}  \brac{\Bar{p}-p} \vect{n} \cdot \Bar{\vect{v}}_{h} \dif s
- \sum_K \int_{\pd K} \frac{\alpha}{h_{K}} 2 \nu \brac{\Bar{\vect{u}}
                 - \vect{u} } \otimes \vect{n} \cdot \Bar{\vect{v}}_h \dif s
\\
+ \sum_{K} \int_{\pd K}  \brac{\brac{\Hat{\vect{u}}-\vect{u}} \cdot \vect{n}}
               \vect{u} \cdot \Bar{\vect{v}}_{h} \dif s
+ \sum_K \int_{\pd K} \lambda \brac{\Hat{\vect{u}} \cdot \vect{n}}
  \brac{\Bar{\vect{u}} - \vect{u}} \cdot \Bar{\vect{v}}_h \dif s
\\
- \brac{1-\chi} \sum_K \int_{\pd K}  \brac{\Hat{\vect{u}} \cdot \vect{n}}
               \Bar{\vect{u}} \cdot \Bar{\vect{v}}_{h} \dif s
-  \int_{\Gamma_{N}} \brac{\chi - \lambda} \brac{\Bar{\vect{u}} \cdot \vect{n}}
    \Bar{\vect{u}} \cdot \Bar{\vect{v}}_{h} \dif s
\\
= \int_{\Gamma_{N}} \vect{h} \cdot \Bar{\vect{v}}_{h}  \dif s
  \quad \forall \ \Bar{\vect{v}} \in \Bar{V}_{h},
\label{eqn:consistency_proof_m_global}
\end{multline}
which holds due to the regularity of~$\vect{u}$,
because $\Bar{p}=\gamma(p)$ and $\Bar{\vect{u}} =
\Hat{\vect{u}} = \gamma(\vect{u})$ and due to satisfaction
of the flux boundary condition~\eqref{eqn:NeumannBC}.
Equation~\eqref{eqn:consistency} follows from the summation of
\eqref{eqn:consistency_proof_c_local}--\eqref{eqn:consistency_proof_m_global}
and the linearity of $F$ in $\vect{v}$, $\Bar{\vect{v}}$, $q$ and
$\Bar{q}$.
\end{proof}

Key to the proof of Proposition~\ref{prop:consistency} is
the consistent formulation of the numerical fluxes, that is
$\Hat{\vect{\sigma}} = \vect{\sigma}$ and $\Hat{\vect{u}} = \vect{u}$
if $\Bar{\vect{u}}=\vect{u}$ and $\Bar{p}=p$.

\begin{proposition}[mass conservation]
\label{prop:mass_conservation}
If $\vect{u}_{h}$ and $\Bar{\vect{u}}_{h}$ satisfy~\eqref{eqn:fe_problem}, then
\begin{equation}
  \int_{\pd K} \Hat{\vect{u}}_{h} \cdot \vect{n} \dif s = 0
    \quad \forall \ K \in \mathcal{T}
\label{eqn:local_mass_conservation}
\end{equation}
and
\begin{equation}
   \int_{\pd \Omega} \Bar{\vect{u}}_{h} \cdot \vect{n}\dif s = 0.
\label{eqn:global_mass_conservation}
\end{equation}
\end{proposition}
\begin{proof}
Setting $\vect{v}_{h} = \Bar{\vect{v}}_{h} = \vect{0}$, $\Bar{q}_{h} = 1$,
and $q_{h} = 1$ on the cell $K$ and $q_{h} = 0$ on $\mathcal{T} \setminus
K$ leads to~\eqref{eqn:local_mass_conservation}.  Setting $\vect{v}_{h}
= \Bar{\vect{v}}_{h} = \vect{0}$ and $q_{h} = \Bar{q}_{h} = 1$
in~\eqref{eqn:fe_problem} leads to~\eqref{eqn:global_mass_conservation}.
\end{proof}

The local conservation property is in terms of the numerical mass flux
$\Hat{\vect{u}}$, as is typical for discontinuous Galerkin methods applied
to Stokes flow \citep{cockburn:2002,cockburn:2009}.  Classical local mass
conservation would be satisfied if $\beta =0$, but $\beta$ must be greater
than zero for stability of the equal-order case~\citep{hughes:1987}.
If the pressure field is chosen to be one polynomial order lower than
the velocity field, then $\beta$ can be set equal to zero and mass is
conserved locally and exactly. However, it will be shown that reducing the
size of the pressure space requires compromising on either momentum
conservation or energy stability.

\begin{proposition}[momentum conservation]
\label{prop:momentum_conservation}
If $\vect{u}_{h}$ and $\Bar{\vect{u}}_{h}$ solve~\eqref{eqn:fe_problem},
and the function spaces $V_{h}$, $\Bar{V}_{h}$, $Q_{h}$ and $\Bar{Q}_{h}$
are selected such that for a constant but otherwise arbitrary vector
$\vect{c}$ it holds that
$\vect{v}_{h} \cdot \vect{c} \in Q_{h} \forall \ \vect{v}_{h} \in
V_{h}$ and $\Bar{\vect{v}}_{h} \cdot \vect{c} \in \Bar{Q}_{h} \forall \
\Bar{\vect{v}}_{h} \in \Bar{V}_{h}$, then
\begin{equation}
  \frac{d}{dt} \int_{K} \vect{u}_{h} \dif x
    = \int_{K} \vect{f} \dif x - \int_{\pd K} \Hat{\vect{\sigma}}_{h} \vect{n} \dif s
  \quad \forall \ K \in \mathcal{T},
\label{eqn:local_momentum_conservation}
\end{equation}
and if $\Gamma_{D} = \emptyset$
\begin{equation}
  \frac{d}{dt}  \int_{\Omega} \vect{u}_{h} \dif x
  = \int_{\Omega} \vect{f} \dif x
  - \int_{\pd \Omega} \brac{1 - \lambda} \brac{\Bar{\vect{u}}_{h} \cdot \vect{n}} \Bar{\vect{u}}_{h} \dif s
  - \int_{\pd \Omega} \vect{h} \dif s.
\label{eqn:global_momentum_conservation}
\end{equation}
\end{proposition}
\begin{proof}
Setting $\vect{v}_{h} = \vect{e}_{j}$ and $q_{h} = - (1-\chi) \vect{u}_{h}
\cdot \vect{e}_{j}$ on $K$, where $\vect{e}_{j}$ is a canonical unit
basis vector, $\vect{v}_{h} = \vect{0}$ and $q_{h} = 0$ on $\mathcal{T}
\backslash K$, $\Bar{\vect{v}}_{h} = \vect{0}$ and $\Bar{q}_{h} = 0$
in~\eqref{eqn:fe_problem},
\begin{multline}
   \frac{d}{dt}\int_{K} \vect{u}_{h} \cdot \vect{e}_{j} \dif x
  + \int_{\pd K} \brac{\Hat{\vect{u}}_{h} \cdot \vect{n}} \vect{u}_{h} \cdot \vect{e}_{j} \dif s
  + \int_{\pd K} \lambda \brac{\Hat{\vect{u}}_{h} \cdot \vect{n}}
             \brac{\Bar{\vect{u}}_{h} - \vect{u}_{h}} \cdot \vect{e}_{j} \dif s
\\
  + \int_{\pd K} \Hat{\vect{\sigma}}_{d, h} \vect{n} \cdot \vect{e}_{j} \dif s
  = \int_{K} \vect{f} \cdot \vect{e}_{j} \dif x,
\end{multline}
which from the definition of the fluxes in equation
\eqref{eqn:adv_int_flux} proves \eqref{eqn:local_momentum_conservation}.

Setting
$\vect{v}_{h} = \vect{e}_{j}$,
$\Bar{\vect{v}}_{h} = -\vect{e}_{j}$,
$q_{h} = -(1 - \chi) \vect{u}_{h} \cdot \vect{e}_{j}$
and
$\Bar{q}_{h} = -(1 - \chi) \Bar{\vect{u}}_{h} \cdot \vect{e}_{j}$
in~\eqref{eqn:fe_problem} leads
to~\eqref{eqn:global_momentum_conservation} directly.
\end{proof}

Local momentum conservation is in terms the numerical flux
$\Hat{\vect{\sigma}}_h$, as is a typical feature of discontinuous Galerkin
methods.  Note also the requirement on the size of the pressure space
relative to the components of the velocity space, which would not be
satisfied by methods that use lower-order polynomials for the pressure
than for the velocity, such as Taylor--Hood elements. Such elements
only conserve momentum when conservative forms of the momentum equation
are used which, in the advective limit, requires compromising on energy
stability. Provided that the requirements on the sizes of the function
spaces are met, momentum conservation holds irrespective of the value of
$\chi$, i.e. for advective as well as conservative forms of the advection
operator, and it will be shown that for $\chi=1/2$ the method is also
energy stable (see proposition~\ref{prop:energy_stability}).

For cases with Dirichlet boundary conditions ($\Gamma_{D} \ne \emptyset$),
demonstration of momentum conservation is less straightforward
since $\Bar{\vect{v}}_{h}$ can not be set equal to $\vect{e}_{j}$
on $\Gamma_{D}$.  This difficulty can be overcome by introducing
auxiliary flux terms on $\Gamma_{D}$. Details of the approach can be
found in Ref.~\citep{Hughes:2005}.

\begin{proposition}[global energy stability]
\label{prop:energy_stability}
If $\vect{u}_{h}$ solves \eqref{eqn:fe_problem} with $\chi = 1/2$
and homogeneous boundary conditions, then in the absence of forcing
terms and for a suitably large $\alpha$
\begin{equation}
   \frac{d}{dt} \int_{\Omega} \left| \vect{u}_{h} \right|^{2} \dif x \le 0.
\label{eqn:energy_stability}
\end{equation}
\end{proposition}
\begin{proof}
Setting
$\vect{v}_{h} = \vect{u}_{h}$,
$\Bar{\vect{v}}_{h} = -\Bar{\vect{u}}_{h}$,
$q_{h} = - p_{h}$
and
$\Bar{q}_{h} = - \Bar{p}_{h}$
in~\eqref{eqn:fe_problem} gives, for $\chi=1/2$,
\begin{multline}
   \int_{\Omega} \frac{\pd \vect{u}_{h}}{\pd t} \cdot \vect{u}_{h} \dif x
  + \sum_{K} \int_{\pd K} \brac{\frac{1}{2}-\lambda} \brac{\Hat{\vect{u}}_{h} \cdot \vect{n}}
             \left| \Bar{\vect{u}}_{h} - \vect{u}_{h} \right|^{2} \dif s
\\
  - \sum_{K}\int_{K} \vect{\sigma}_{d, h} \colon \nabla \vect{u}_{h} \dif x
  - \sum_{K} \int_{\pd K} \Hat{\vect{\sigma}}_{d, h} \vect{n} \cdot \brac{\Bar{\vect{u}}_{h} - \vect{u}_{h}} \dif s
\\
  + \sum_{K} \int_{\pd K} 2\nu \brac{\Bar{\vect{u}}_{h}-\vect{u}_{h}} \cdot \nabla^{s} \vect{u}_{h} \vect{n} \dif s
+  \int_{\Gamma_{N}} \brac{\frac{1}{2} - \lambda} \brac{\vect{\Bar{u}}_{h} \cdot \vect{n}}
             \left| \Bar{\vect{u}}_{h} \right|^{2}\dif s
\\
 - \sum_{K} \int_{K} \vect{u}_{h} \cdot \nabla p_{h} \dif x
 - \sum_{K} \int_{\pd K} \Hat{\vect{u}}_{h} \cdot \vect{n} \brac{\Bar{p}_{h}-p_{h}} \dif s
 + \int_{\pd \Omega} \Bar{\vect{u}}_{h} \cdot \vect{n} \Bar{p}_{h} \dif s
  = 0.
\end{multline}
Substituting the expressions for the diffusive fluxes given in
\eqref{eqn:def_mom_flux} and~\eqref{eqn:diff_int_flux} into the preceding equation,
\begin{multline}
   \int_{\Omega}  \frac{1}{2} \frac{\pd \left| \vect{u}_{h} \right|^{2}}{\pd t} \dif x
  + \sum_{K} \int_{\pd K} \frac{1}{2}\left|\Hat{\vect{u}}_{h} \cdot \vect{n}\right|
             \left| \Bar{\vect{u}}_{h} - \vect{u}_{h} \right|^{2} \dif s
  + \sum_{K}\int_{K} 2 \nu \left| \nabla^{s} \vect{u}_{h} \right|^{2} \dif x
\\
  + \sum_{K} \int_{\pd K} \frac{\alpha}{h_K} 2 \nu \left| \Bar{\vect{u}}_{h} - \vect{u}_{h} \right|^{2} \dif s
  + 2 \sum_{K} \int_{\pd K} 2\nu \brac{\Bar{\vect{u}}_{h} - \vect{u}_{h}}
    \cdot \nabla^{s} \vect{u} \vect{n} \dif s
\\
  - \sum_{K}\int_{K} p_{h} \vect{I} \colon \nabla \vect{u}_{h} \dif x
  - \sum_{K}\int_{\pd K} \Bar{p}_h \vect{n} \cdot \brac{\Bar{\vect{u}}_h-\vect{u}_h}\dif s
+  \int_{\Gamma_{N}} \frac{1}{2}\left| \Bar{\vect{u}}_{h} \cdot \vect{n} \right|
             \left| \Bar{\vect{u}}_{h} \right|^{2}\dif s
\\
 - \sum_{K} \int_{K} \vect{u}_{h} \cdot \nabla p_{h} \dif x
 - \sum_{K} \int_{\pd K} \Hat{\vect{u}}_{h} \cdot \vect{n} \brac{\Bar{p}_{h}-p_{h}} \dif s
 + \int_{\pd \Omega} \Bar{\vect{u}}_{h} \cdot \vect{n} \Bar{p}_{h} \dif s
  = 0,
\label{eqn:energy_1}
\end{multline}
in which we have also used $\brac{1/2-\lambda}\vect{u}\cdot\vect{n} =
\left| \vect{u} \cdot \vect{n}\right|/2$ on facets. After substitution of the
mass flux $\Hat{\vect{u}}_{h}$ given in~\eqref{eqn:def_uhat} and
the application of partial integration of the pressure gradient term, we finally obtain,
\begin{multline}
  \frac{d}{dt} \int_{\Omega}  \frac{1}{2} \left| \vect{u}_{h} \right|^{2} \dif x
  + \sum_{K} \int_{\pd K} \frac{1}{2}\left| \Hat{\vect{u}}_{h} \cdot \vect{n}\right|
             \left| \Bar{\vect{u}}_{h} - \vect{u}_{h} \right|^{2} \dif s
  + \sum_{K}\int_{K} 2 \nu \left| \nabla^{s} \vect{u}_{h} \right|^{2} \dif x
\\
  + \sum_{K} \int_{\pd K} \frac{\alpha}{h_K} 2 \nu \left| \Bar{\vect{u}}_{h} - \vect{u}_{h} \right|^{2} \dif s
  + 2 \sum_{K} \int_{\pd K} 2\nu \brac{\Bar{\vect{u}}_{h} - \vect{u}_{h}}
    \cdot \nabla^{s} \vect{u}_{h}  \vect{n} \dif s
\\
+  \int_{\Gamma_{N}} \frac{1}{2}\left| \vect{\Bar{u}}_{h} \cdot \vect{n} \right|
             \left| \Bar{\vect{u}}_{h} \right|^{2}\dif s
  + \sum_{K}\int_{\pd K} \frac{\beta h_{K}}{1+\nu} \left| \Bar{p}_{h} - p_{h} \right|^{2} \dif s
  = 0.
\label{eqn:energy_2}
\end{multline}
For the case $\nu = 0$, all terms in~\eqref{eqn:energy_2} other
than the time derivative term, are guaranteed to be non-negative,
and therefore~\eqref{eqn:energy_stability} holds. For the case $\nu > 0$,
no conclusion as to the sign of $\brac{\Bar{\vect{u}}_h-\vect{u}_h}
\cdot \nabla^{s} \vect{u}_h \vect{n}$ in~\eqref{eqn:energy_2} can
be drawn.  However, there exists an $\alpha > 0$, independent of
$h_{K}$, such that~\eqref{eqn:energy_stability} holds (see the proof in
\citep{Wells:2010} for the diffusion equation).
\end{proof}

The key to the energy stability is the use of the combined
conservative/advective form of the momentum equations.  The total kinetic
energy in the method will decrease monotonically, despite $\vect{u}_{h}$
not being point-wise divergence-free. The amount of dissipation is
determined by the difference between the cell and facet velocity fields on
facets ($\Bar{\vect{u}}_{h} - \vect{u}_{h}$) and the difference between
the cell and facet pressure fields on facets ($\Bar{p}_{h} - p_{h}$). It
is not possible to prove a cell-wise kinetic energy inequality.

Recall that the energy stable scheme is also momentum
conserving if the dimensional components of the
velocity space are subspaces of the pressure space (see
Proposition~\ref{prop:momentum_conservation}).  Otherwise, simultaneous
momentum conservation (\eqref{eqn:local_momentum_conservation}
and~\eqref{eqn:global_momentum_conservation}) and energy
stability~\eqref{eqn:energy_stability} is not possible.  Notably, this
will be the case for finite element methods using lower-order polynomials
for the pressure than for the velocity. For advection dominated flows
such elements can only guarantee energy stability when compromising on
momentum conservation or by adding artificial viscosity. In the viscous
limit the requirement on the size of the velocity and pressure spaces
can be relaxed without compromising energy stability, permitting a lower
polynomial order for the pressure space than for the velocity space,
which for Stokes flow is advantageous from the viewpoint of accuracy,
as will be shown in Section~\ref{sec:stokes-with-source}.

A complete stability proof, for the Stokes case alone, would require a
more subtle analysis, with the usual stability conditions demonstrated
for suitably defined norms that include both functions on cells and the
functions on facets. \emph{A priori} stability and convergence estimates
for the method applied to the advection-reaction-diffusion equation have
been proved~\citep{Wells:2010}, and efforts in this direction for the
Stokes equations are ongoing.

\section{A fully-discrete formulation}

We now present a fully-discrete formulation.  The time interval of
interest, $I$, is partitioned such that $I = \left(0, t_1, \hdots,t_{N-1},
t_{N} \right]$ and time increments are denoted $\delta t_n = t_{n+1} -
t_n$.  We consider a $\theta$-method for dealing with the time derivative,
with mid-point values of a function $y$ given by
\begin{equation}
  y_{n+\theta} \colonequals \brac{1-\theta} y_{n} + \theta y_{n+1},
\end{equation}
where $\theta \in [0, 1]$ is a parameter.

The advective velocity will be evaluated at the current time
$t_{n}$, thereby linearizing the problem (Picard linearization).
For the momentum-related $F$-functionals presented in
Section~\ref{sec:combined_method}, we now present time-discrete
counterparts. The term $\lambda$ is always evaluated on the basis of
the known velocity field at time $t_{n}$.  A time-discrete counterpart
of~\eqref{eqn:local_weak_momentum_L} reads
\begin{multline}
 F_{\delta t, \rm m}(\vect{U}_{n+1}; \vect{v})
  \colonequals
     \int_{\Omega} \frac{\vect{u}_{n+1} - \vect{u}_{n}}{\delta t} \cdot \vect{v} \dif x
  - \chi \sum_{K}\int_{K} \brac{\vect{u}_{n+\theta} \otimes \vect{u}_{n}} \colon \nabla \vect{v} \dif x
\\
  + \brac{1-\chi} \sum_{K}\int_{K} \brac{\nabla \vect{u}_{n+\theta} \vect{u}_{n} }\cdot \vect{v} \dif x
  + \chi \sum_{K} \int_{\pd K} \brac{\Hat{\vect{u}}_{n} \cdot \vect{n}}  \vect{u}_{n+\theta} \cdot \vect{v} \dif s
\\
  + \sum_{K} \int_{\pd K} \lambda \brac{\Hat{\vect{u}}_{n} \cdot \vect{n}}
                  \brac{\Bar{\vect{u}}_{n+\theta} - \vect{u}_{n+\theta}} \cdot \vect{v} \dif s
  - \sum_{K}\int_{K} \vect{\sigma}_{d, {n+\theta}} \colon \nabla \vect{v} \dif x
\\
  + \sum_{K} \int_{\pd K} \Hat{\vect{\sigma}}_{d, {n+\theta}} \vect{n} \cdot \vect{v} \dif s
  + \sum_{K} \int_{\pd K} 2\nu \brac{\Bar{\vect{u}}_{n+\theta} - \vect{u}_{n+\theta}}
    \cdot \nabla^{s} \vect{v} \vect{n} \dif s
\\
  - \int_{\Omega} \vect{f}_{n+\theta} \cdot \vect{v} \dif x,
\end{multline}
and a time-discrete counterpart of~\eqref{eqn:global_weak_momentum_L}
reads
\begin{multline}
 \Bar{F}_{\delta t, \rm m}(\vect{U}_{n+1}; \Bar{\vect{v}})
    \colonequals
 \chi  \sum_{K} \int_{\pd K}  \brac{\Hat{\vect{u}}_{n} \cdot \vect{n}}
               \vect{u}_{n+\theta} \cdot \Bar{\vect{v}} \dif s
\\
- \brac{1-\chi} \sum_{K} \int_{\pd K}  \brac{\Hat{\vect{u}}_{n} \cdot \vect{n}}
               \brac{\Bar{\vect{u}}_{n+\theta} - \vect{u}_{n+\theta}} \cdot \Bar{\vect{v}} \dif s
\\
+ \sum_{K} \int_{\pd K}  \lambda \brac{\Hat{\vect{u}}_{n} \cdot \vect{n}}
    \brac{\Bar{\vect{u}}_{n+\theta} - \vect{u}_{n+\theta}} \cdot \Bar{\vect{v}} \dif s
+  \sum_{K} \int_{\pd K} \Hat{\vect{\sigma}}_{d, n+\theta} \vect{n} \cdot \Bar{\vect{v}} \dif s
\\
-  \int_{\Gamma_{N}} \brac{\chi - \lambda} \brac{\Bar{\vect{u}}_{n} \cdot \vect{n}}
    \Bar{\vect{u}}_{n+\theta} \cdot \Bar{\vect{v}} \dif s
- \int_{\Gamma_{N}} \vect{h}_{n+\theta} \cdot \Bar{\vect{v}}  \dif s.
\end{multline}
Defining
\begin{multline}
F_{\delta t}(\vect{U}_{n+1}; \vect{W})
\colonequals
  F_{\delta t, \rm m}(\vect{U}_{n+1}; \vect{v})
\\
+ \Bar{F}_{\delta t, \rm m}(\vect{U}_{n+1}; \Bar{\vect{v}})
  + F_{\rm c}(\vect{U}_{n+1}; q)
  + \Bar{F}_{\rm c}(\vect{U}_{n+1}; \Bar{q}),
\end{multline}
a fully-discrete finite element problem at time $t_{n+1}$ involves:
given the solution
$\vect{U}_{h,n} \in  V_{h} \times \Bar{V}_{h} \times Q_{h} \times \Bar{Q}_{h}$
at time $t_n$,
the forcing term $\vect{f}_{n+\theta} \in \left[ L^2\brac{\Omega} \right]^d$,
the boundary condition
$\vect{h}_{n+\theta} \in \left[ L^2\brac{\Gamma_N} \right]^d$
and the viscosity $\nu$,
find $\vect{U}_{h, n+1} \in V_{h} \times \Bar{V}_{h} \times Q_{h} \times \Bar{Q}_{h}$
such that
\begin{equation}
F_{\delta t}(\vect{U}_{h,n+1}; \vect{W}_{h}) = 0
    \quad \forall \ \vect{W}_{h} \in V_{h} \times \Bar{V}_{h} \times Q_{h} \times \Bar{Q}_{h}.
\label{eqn:fe_problem_discrete}
\end{equation}
For the scheme that has been adopted, $F_{\delta t}$ is linear in both
$\vect{U}_{h, n+1}$ and $\vect{W}_{h}$.

We now demonstrate that the considered fully discrete formulation
inherits the conservation and energy stability properties of the
semi-discrete case. As for the semi-discrete case, all results
hold if the spaces $\Bar{V}_{h}$ and $\Bar{Q}_{h}$ are replaced
by $\Bar{V}^{\star}_{h}$ and $\Bar{Q}^{\star}_{h}$, respectively.

\begin{proposition}[fully discrete mass conservation]
If $\vect{u}_{h, n+1}$ and $\Bar{\vect{u}}_{h, n+1}$
satisfy~\eqref{eqn:fe_problem_discrete}, then
\begin{equation}
  \int_{\pd K} \Hat{\vect{u}}_{h, n+1} \cdot \vect{n} \dif s = 0
  \quad \forall \ K \in \mathcal{T},
\label{eqn:local_mass_conservation_discrete}
\end{equation}
and
\begin{equation}
   \int_{\pd \Omega} \Bar{\vect{u}}_{h, n+1} \cdot \vect{n}\dif s = 0.
\label{eqn:global_mass_conservation_discrete}
\end{equation}
\end{proposition}
\begin{proof}
The proof follows the same steps as the proof of
Proposition~\ref{prop:mass_conservation}.
\end{proof}

\begin{proposition}[fully discrete momentum conservation]
If $\vect{U}_{h, n}$ and $\vect{U}_{h, n+1}$
solve~\eqref{eqn:fe_problem_discrete}, and the function spaces
$V_{h}$, $\Bar{V}_{h}$, $Q_{h}$ and $\Bar{Q}_{h}$ are selected such
that for a constant but otherwise arbitrary vector $\vect{c}$
it holds that $\vect{v}_{h} \cdot \vect{c} \in Q_{h} \forall \ \vect{v}_{h} \in
V_{h}$ and $\Bar{\vect{v}}_{h} \cdot \vect{c} \in \Bar{Q}_{h} \forall \
\Bar{\vect{v}}_{h} \in \Bar{V}_{h}$, then
\begin{equation}
  \int_{K} \frac{\vect{u}_{h, n+1} - \vect{u}_{h, n}}{\delta t} \dif x
    = \int_{K} \vect{f}_{n+\theta} \dif x
     - \int_{\pd K} \Hat{\vect{\sigma}}_{h, n+\theta} \vect{n} \dif s
\quad \forall \ K \in \mathcal{T},
\label{eqn:momentum_conservation_local_discrete}
\end{equation}
and if $\Gamma_{D} = \emptyset$
\begin{multline}
  \int_{\Omega} \frac{\vect{u}_{h, n+1} - \vect{u}_{h, n}}{\delta t}  \dif x
  = \int_{\Omega} \vect{f}_{n+\theta} \dif x
\\
  - \int_{\pd \Omega} \brac{1 - \lambda} \brac{\Bar{\vect{u}}_{h, n}
   \cdot \vect{n}} \Bar{\vect{u}}_{h, n+\theta} \dif s
  - \int_{\pd \Omega} \vect{h}_{n+\theta} \dif s.
\label{eqn:momentum_conservation_global_discrete}
\end{multline}
\end{proposition}
\begin{proof}
The proof follows the same steps as the proof
to Proposition~\ref{prop:momentum_conservation}.
Equation~\eqref{eqn:momentum_conservation_local_discrete} follows from
setting $\vect{v}_{h} = \vect{e}_{j}$ and $q_{h} = - (1-\chi) \vect{u}_{h,
n+\theta} \cdot \vect{e}_{j}$ on $K$, $\vect{v}_{h} = \vect{0}$ and $q_{h}
= 0$ on $\mathcal{T} \backslash K$, $\Bar{\vect{v}}_{h} = \vect{0}$
and $\Bar{q}_{h} = 0$ in~\eqref{eqn:fe_problem_discrete}.
Equation~\eqref{eqn:momentum_conservation_global_discrete} follows
from setting
$\vect{v}_{h} = \vect{e}_{j}$,
$\Bar{\vect{v}}_{h} = -\vect{e}_{j}$,
$q_{h} = -(1 - \chi) \vect{u}_{h, n+\theta} \cdot \vect{e}_{j}$
and
$\Bar{q}_{h} = -(1 - \chi) \Bar{\vect{u}}_{h, n+\theta} \cdot \vect{e}_{j}$
in~\eqref{eqn:fe_problem_discrete}.
\end{proof}

\begin{proposition}[fully discrete energy stability]
If $\vect{U}_{h, n}$ and $\vect{U}_{h, n+1}$ solve
\eqref{eqn:fe_problem_discrete} with $\chi = 1/2$, $\vect{f} = \vect{0}$
and $\Gamma_{D} = \pd \Omega$, for $\theta \ge 1/2$ and suitably large
$\alpha$
\begin{equation}
   \int_{\Omega} \left| \vect{u}_{h,n+1} \right|^{2} \dif x
      \le \int_{\Omega} \left| \vect{u}_{h,n} \right|^{2} \dif x.
\label{eqn:discrete_energy_stability}
\end{equation}
\end{proposition}
\begin{proof}
For $\chi=1/2$, setting
$\vect{v}_{h} = \vect{u}_{h, n+\theta}$,
$\Bar{\vect{v}}_{h} = -\Bar{\vect{u}}_{h, n+\theta}$
$q_{h} = -\theta p_{h, n+\theta}$
and
$\Bar{q}_{h} = - \theta \Bar{p}_{h, n+\theta}$
in~\eqref{eqn:fe_problem_discrete}, and adding to \eqref{eqn:fe_problem_discrete}
$F_{c}(\vect{U}_{n}; -(1-\theta) p_{n+\theta})$
and $\Bar{F}_{c}(\vect{U}_{n}; -(1-\theta) p_{n+\theta})$,
\begin{multline}
   \int_{\Omega} \frac{\vect{u}_{h, n+1} - \vect{u}_{h, n}}{\delta t} \cdot \vect{u}_{h, n+\theta} \dif x
  + \sum_{K} \int_{\pd K} \brac{\frac{1}{2}-\lambda} \brac{\Hat{\vect{u}}_{h, n} \cdot \vect{n}}
             \left| \Bar{\vect{u}}_{h, n+\theta} - \vect{u}_{h, n+\theta} \right|^{2} \dif s
\\
  - \sum_{K}\int_{K} \vect{\sigma}_{d, h, n+\theta} \colon \nabla \vect{u}_{h, n+\theta} \dif x
  - \sum_{K} \int_{\pd K} \Hat{\vect{\sigma}}_{d, h, n+\theta} \vect{n}
        \cdot \brac{\Bar{\vect{u}}_{h, n+\theta} - \vect{u}_{h, n+\theta}} \dif s
\\
  + \sum_{K} \int_{\pd K} 2\nu \brac{\Bar{\vect{u}}_{h, n+\theta}
                  - \vect{u}_{h, n+\theta}} \cdot \nabla^{s} \vect{u}_{h, n+\theta} \vect{n} \dif s
\\
+  \int_{\Gamma_{N}} \brac{\frac{1}{2} - \lambda} \brac{\vect{\Bar{u}}_{h, n} \cdot \vect{n}}
             \left| \Bar{\vect{u}}_{h, n+\theta} \right|^{2}\dif s
 - \sum_{K} \int_{K} \vect{u}_{h, n+\theta} \cdot \nabla p_{h, n+\theta} \dif x
\\
 - \sum_{K} \int_{\pd K} \Hat{\vect{u}}_{h, n+\theta} \cdot \vect{n} \brac{\Bar{p}_{h, n+\theta} - p_{h, n+\theta}} \dif s
 + \int_{\pd \Omega} \Bar{\vect{u}}_{h, n+\theta} \cdot \vect{n} \Bar{p}_{h, n+\theta} \dif s
  = 0.
\end{multline}
Taking into account that
\begin{equation}
  \vect{u}_{h, n+\theta} = \delta t \brac{\theta - \frac{1}{2}}
  \frac{\vect{u}_{h, n+1} - \vect{u}_{h, n}}{\delta t}
+  \frac{\vect{u}_{h, n+1} + \vect{u}_{h, n}}{2}
\end{equation}
leads to
\begin{multline}
  \int_{\Omega} \frac{\vect{u}_{h, n+1} - \vect{u}_{h, n}}{\delta t} \cdot \vect{u}_{h, n+\theta} \dif x
  = \brac{\theta - \frac{1}{2}} \int_{\Omega} \frac{| \vect{u}_{h, n+1} - \vect{u}_{h, n} |^{2}}{\delta t} \dif x
\\
  + \frac{1}{2} \int_{\Omega} \frac{| \vect{u}_{h, n+1} |^{2}}{\delta t} \dif x
  - \frac{1}{2} \int_{\Omega} \frac{| \vect{u}_{h, n} |^{2}}{\delta t} \dif x.
\end{multline}
Following the same steps as in Proposition~\ref{prop:energy_stability}
proves that \eqref{eqn:discrete_energy_stability} holds when $\theta
\ge 1/2$ and for sufficiently large $\alpha$.
\end{proof}

\section{Algorithmic aspects}
\label{sec:algorithmic-aspects}

The fully-discrete finite element problem
in~\eqref{eqn:fe_problem_discrete} can lead to an efficient numerical
implementation in which the functions on cells ($\vect{u}_{h, n+1}$ and
$p_{h, n+1}$) are eliminated cell-wise in favor of the functions that live
only on cell facets ($\Bar{\vect{u}}_{h, n+1}$ and $\Bar{p}_{h, n+1}$) via
static condensation.  Key to this algorithmic feature is
that functions on cells are not linked directly across cell facets,
in contrast with conventional discontinuous Galerkin methods. Rather,
functions on neighboring cells communicate via the functions defined
only on cell facets. Moreover, if the functions $\Bar{\vect{u}}_{h}$ and
$\Bar{p}_{h}$ are chosen to be continuous, the method will result in the
same number of global degrees of freedom as for a continuous method on the
same mesh (if interior degrees of freedom are eliminated from a continuous
method via static condensation). Despite having the same number of global
degrees offreedom as a continuous method, stabilizing mechanisms that are
typical of discontinuous Galerkin methods are naturally incorporated.
Moreover, for the advection-diffusion equation, it has been proved that the
approach has the same stability properties as classical upwinded discontinuous
Galerkin methods~\citep{Wells:2010}.  More detailed discussions on
algorithmic aspects can be found in Refs.~\citep{Labeur:2007,Labeur:2009}.

\section{Examples}
\label{sec:examples}

We present now a number of examples in support of the analysis
presented in the preceding sections.  The computer code used to
compute the examples presented in this section has been generated
automatically from expressive input using tools from the FEniCS
Project (\url{http://www.fenicsproject.org}) \citep{logg:2009}.
Specifically, an expressive domain-specific language for finite element
variational statements in combination with automated code generation
\citep{kirby:2006,oelgaard:2010,oelgaard:2008} and a programmable
environment has been used \citep{logg:2009}.  The computer code used
for all examples presented in this work is available under a GNU
public license in the supporting material~\citep{supporting_material}.

All examples use triangular elements, uniform partitionings and
continuous facet functions, that is
$\Bar{\vect{u}}_{h} \in \Bar{V}_{h}^{\star}$ and $\Bar{p}_{h}
\in \Bar{Q}_{h}^{\star}$.
When computing errors, analytical solutions which are polynomial are
represented exactly. Otherwise analytical solutions are interpolated
using eighth-order Lagrange basis functions on the same mesh. Exact
quadrature is used in all cases.
\subsection{Stokes flow with source}
\label{sec:stokes-with-source}

We consider a Stokes problem (by neglecting the momentum advection
terms) with $\nu = 1$ on a unit square with $\vect{u} = \vect{0}$ on
$\pd \Omega$. The source term $\vect{f}$ is chosen such that the exact
solution is:
\begin{equation}
  \begin{split}
    u_x &=  x^2\brac{1 - x}^2\brac{2y - 6y^2 + 4y^3},
    \\
    u_y &= -y^2\brac{1 - y}^2\brac{2x - 6x^2 + 4x^3},
    \\
    p   &= x\brac{1 - x}.
  \end{split}
\label{eqn:stokes-exact}
\end{equation}
The constraint $\int_\Omega p \dif x =1/6$ is
enforced by means of a Lagrange-multiplier, matching the solution
in~\eqref{eqn:stokes-exact}.

We investigate convergence rates in the $L^{2}$-norm for the pressure
and the velocity fields using equal-order polynomial elements
($k = \Bar{k} = m = \Bar{m}$). Polynomial orders ranging from one
to five are considered. These results complement those presented
in \citet{Labeur:2007} for the same boundary-value problem, but
in which the pressure field was continuous and only a polynomial
order of one was considered.  We set $\alpha=6 k^2$, based on
observations for higher-order elements~\citep{Wells:2010}, and use
$\beta = 10^{-4}$.  The observed convergence behavior is presented in
Figure~\ref{fig:stokes_u-p}. Standard convergence rates of order $k+1$ for
the velocity field and of order $k$ for the pressure field are observed.
\begin{figure}
  \center
  \begin{tabular}{c}
  \includegraphics[width=0.75\textwidth]{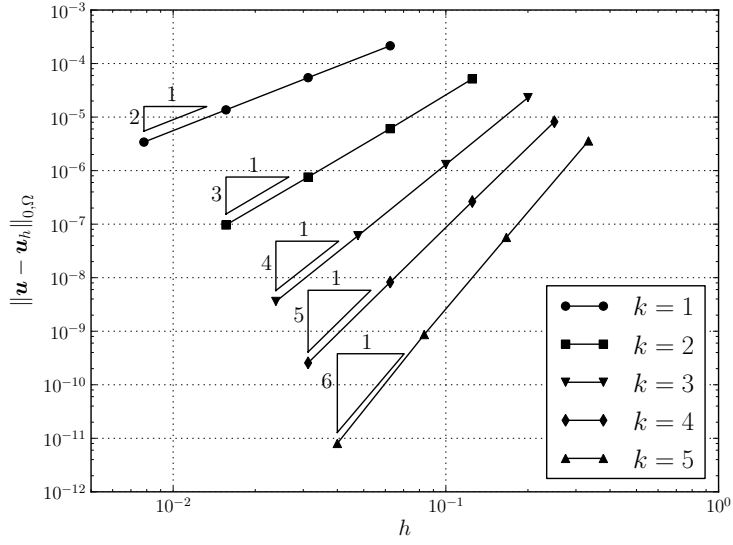}
  \\ (a) \\[1ex]
  \includegraphics[width=0.75\textwidth]{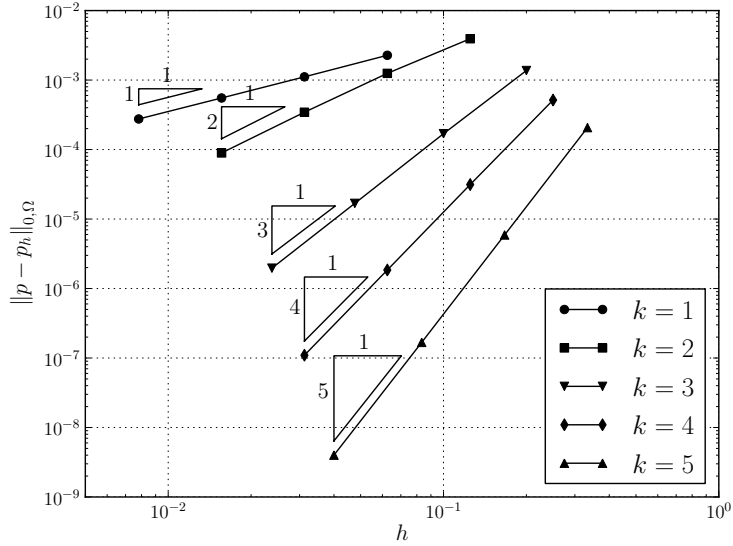}
  \\ (b)
 \end{tabular}
\caption{Stokes flow: computed $L^{2}$ errors in (a) velocity and (b) pressure
         with $h$-refinement and for various polynomial orders $k$
         ($\alpha=6k^2$ and $\beta=10^{-4}$).}
\label{fig:stokes_u-p}
\end{figure}
The error in the divergence of the velocity field is examined via
\begin{equation}
  e_{\rm div}
  \colonequals
  \brac{\sum_K \int_{K} \brac{\nabla \cdot \vect{u}}^2 \dif x }^{1/2},
\label{eqn:div-norm}
\end{equation}
and the computed $e_{\rm div}$ is shown in
Figure~\ref{fig:convergence-run-div} for various polynomial orders and
$h$-refinement.
\begin{figure}
  \center
  \includegraphics[width=0.75\textwidth]{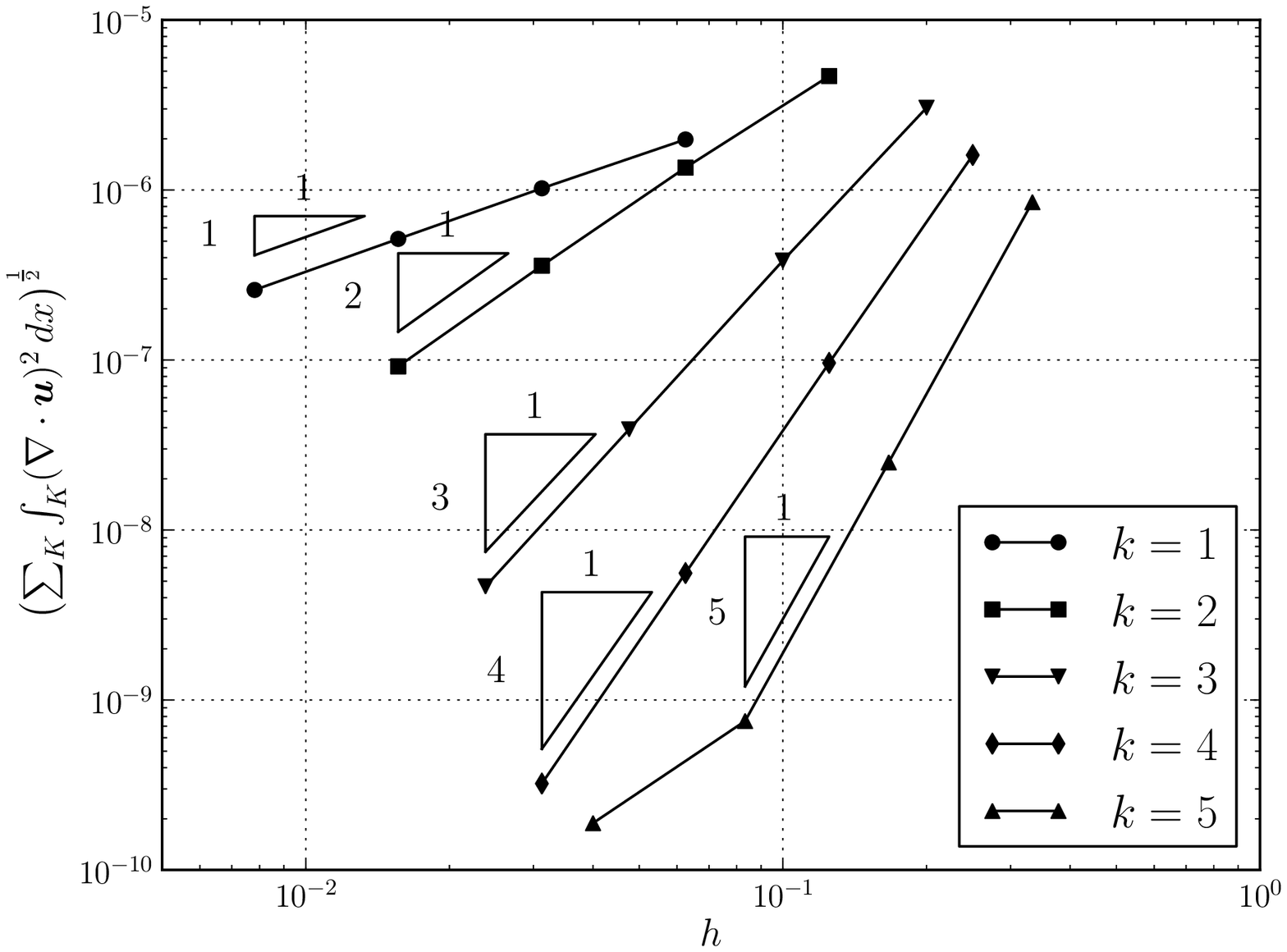}
\caption{Stokes flow with discontinuous pressure: divergence error
      with $h$-refinement and various polynomial orders $k$
      ($\alpha = 6k^2$ and $\beta = 10^{-4}$).}
\label{fig:convergence-run-div}
\end{figure}
Clearly, the divergence error is small.  For comparison, the
divergence errors using the same method for the velocity field, but
with a continuous pressure field \citep{Labeur:2007} are shown in
Figure~\ref{fig:convergence-run-div-cont}.  The observed convergence
rates are similar to those for the discontinuous pressure case.  However,
particularly for the lower-order elements, the divergence error is
significantly greater in the continuous pressure case.
\begin{figure}
  \center
  \includegraphics[width=0.75\textwidth]{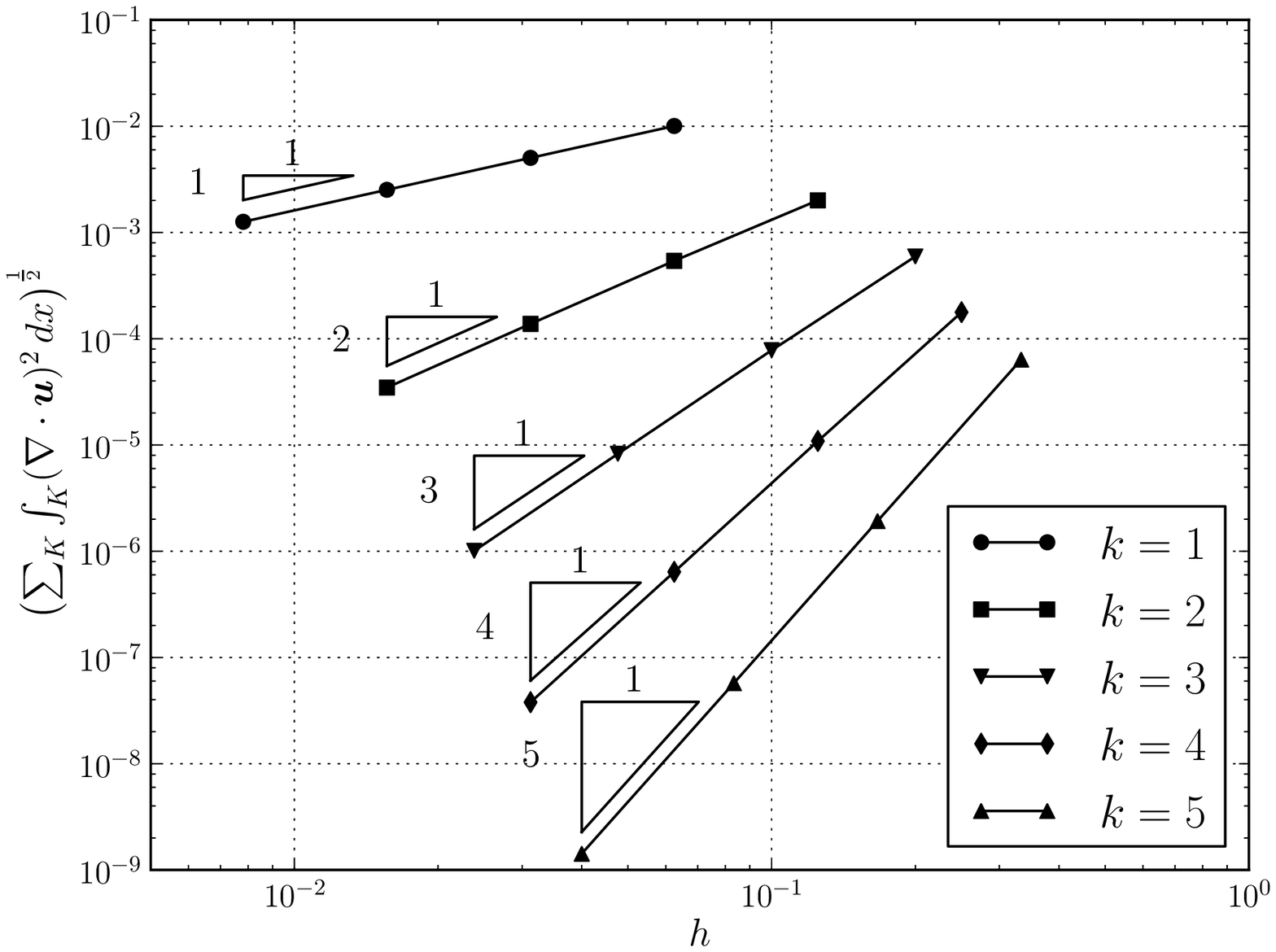}
\caption{Stokes flow with continuous pressure: divergence error
      with $h$-refinement and various polynomial orders $k$
      ($\alpha=8k^2$).}
\label{fig:convergence-run-div-cont}
\end{figure}

For comparison, we consider a Taylor--Hood element with a
continuous piecewise-quadratic velocity field and a continuous
piecewise-linear pressure field, and an element constituted of a
continuous piecewise-quadratic velocity field, enriched by cubic bubble
functions, and a discontinuous piecewise-linear pressure field. The latter
approach is referred to by some authors as the Crouzeix--Raviart method
(e.g.~\citep{Gresho:1998,Elman:2005}), a convention that we adopt here.
For our method we use corresponding polynomial orders of $k =\Bar{k} =
2$ for the velocity and $m = \Bar{m} =1$ for the pressure and penalty
parameters $\alpha=6 k^2$ and~$\beta = 0$. Recall that it is permitted
to use $\beta = 0$ in this case since the polynomial degree of the
pressure field is lower than the polynomial degree of the velocity field,
see Section~\ref{sec:cont_eq}.  The observed convergence
behavior is presented in Figure~\ref{fig:stokes_u-p-combined}, which
shows the expected convergence rates for the Taylor--Hood and
Crouzeix--Raviart methods, and with the method formulated in this
work showing the same rates, which are also the same as
for the $k = \Bar{k} = m = \Bar{m} = 2$ case presented
in Figure~\ref{fig:stokes_u-p}.
\begin{figure}
  \center
  \begin{tabular}{c}
  \includegraphics[width=0.75\textwidth]{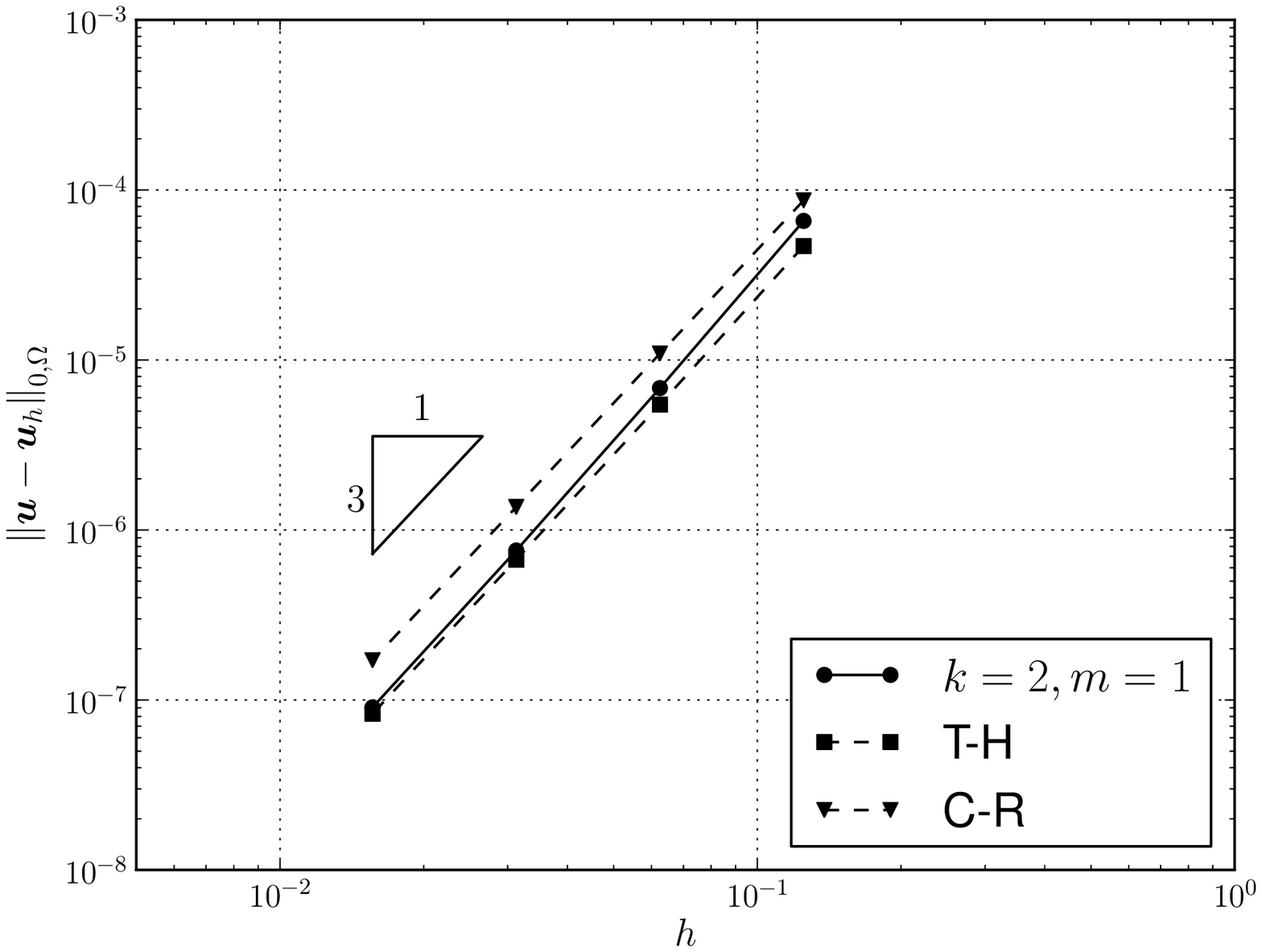}
  \\ (a) \\[1ex]
  \includegraphics[width=0.75\textwidth]{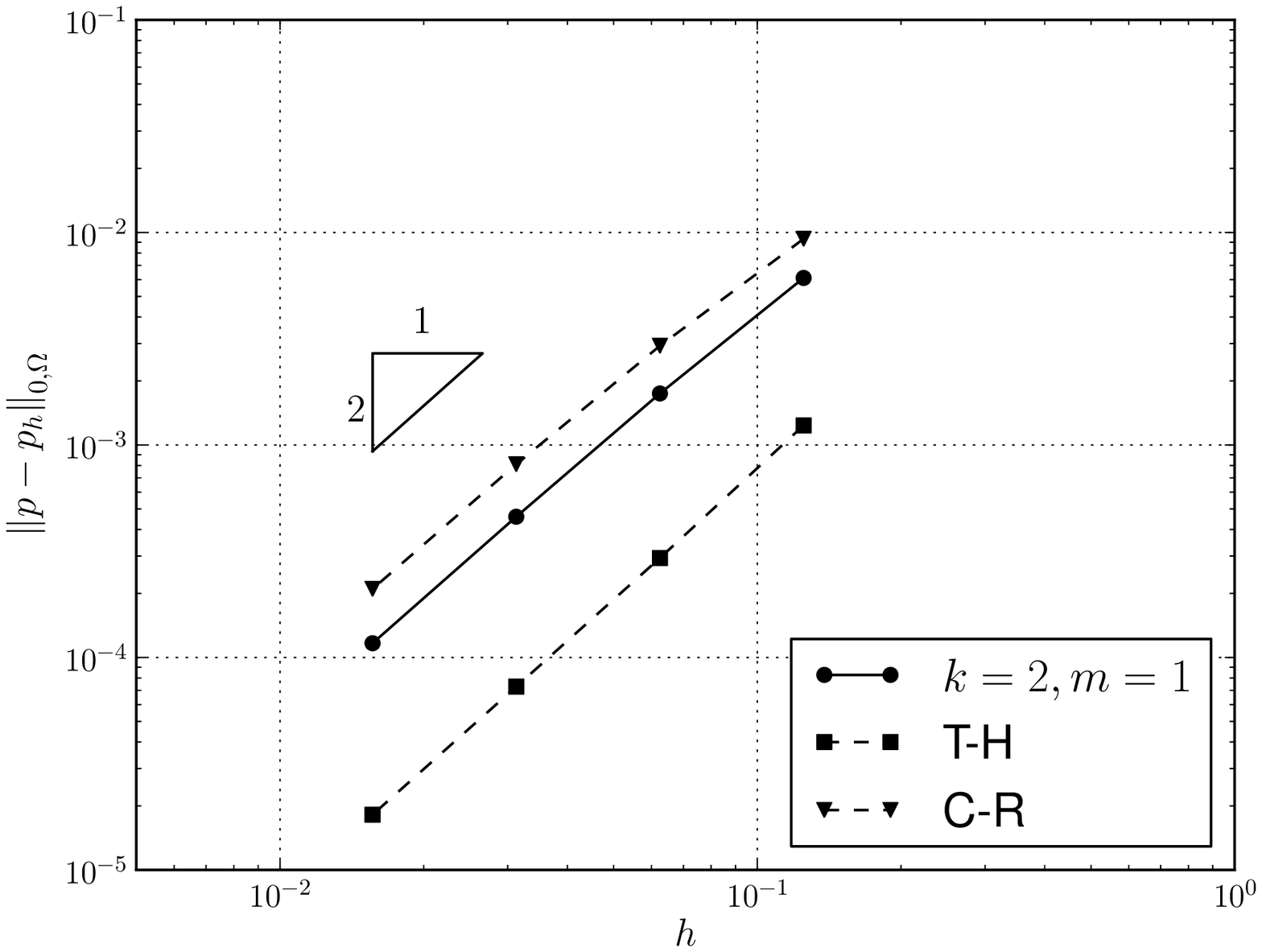}
  \\ (b)
 \end{tabular}
\caption{Stokes flow: computed $L^{2}$ errors in (a) velocity and (b) pressure
         with $h$-refinement for polynomial orders $k=\Bar{k}=2$ and $m=\Bar{m}=1$
         ($\alpha = 6k^2$ and $\beta=0$), compared with Taylor--Hood (T-H) and
         Crouzeix--Raviart (C-R) methods.}
\label{fig:stokes_u-p-combined}
\end{figure}
The divergence error, measured by $e_{\rm div}$, is shown in
Figure~\ref{fig:stokes-div-error-compare}.
\begin{figure}
  \center
  \includegraphics[width=0.75\textwidth]{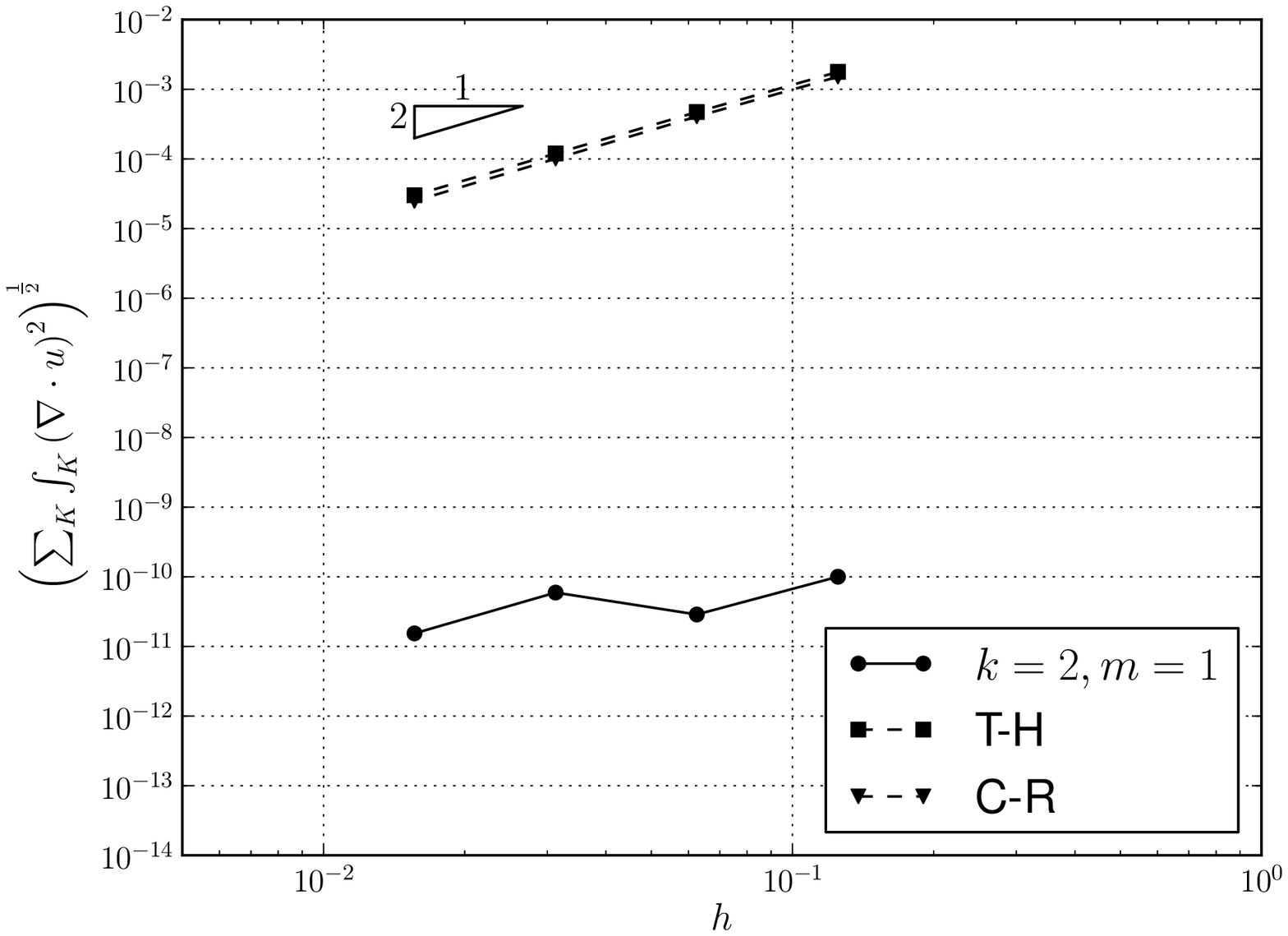}
\caption{Stokes flow: divergence error with $h$-refinement
         for polynomial orders $k=\Bar{k}=2$ and $m=\Bar{m}=1$
         ($\alpha=6k^2$ and $\beta=0$), compared with Taylor--Hood (T-H) and
         Crouzeix--Raviart (C-R) methods.}
\label{fig:stokes-div-error-compare}
\end{figure}
We note that while the Crouzeix--Raviart method conserves mass locally,
the divergence error when measured in $e_{\rm div}$ is very close to that
of the Taylor--Hood method, whereas for our method the divergence error
is effectively zero.

These convergence results demonstrate that in the viscous limit the
equal-order pressure approximation is sub-optimal from the viewpoint
of accuracy. However, in advection dominated flows the simultaneous
satisfaction of Proposition~\ref{prop:momentum_conservation} (momentum
conservation) and Proposition~\ref{prop:energy_stability} (energy
stability) relies on a sufficiently rich pressure field relative
to the velocity field, i.e. $k \leq m$, $\Bar{k} \leq \Bar{m}$, and
consequently~$\beta > 0$.

\subsection{Kovasznay flow}

We now consider the incompressible Navier--Stokes equations by examining
the following analytical solution due to \citet{Kovasznay:1948}:
\begin{equation}
  \begin{split}
    u_x &= 1 - e^{\lambda x} \cos \brac{2\pi y},
    \\
    u_y &=  \frac{\lambda}{2 \pi} e^{\lambda x} \sin \brac{2\pi y},
    \\
    p   &= \frac{1}{2} \brac{1- e^{2 \lambda x}} + C,
  \end{split}
\label{eqn:kovasnay-exact}
\end{equation}
where $C$ is an arbitrary constant and
\begin{equation}
 \lambda = \frac{Re}{2} - \brac{\frac{Re^2}{4} + 4 \pi^2}^{1/2},
\end{equation}
where $Re$ is the Reynolds number. The solution represents laminar flow
in the wake of a grid, see also~\citep{Kirby:2006b,Warburton:2000}.

We use a rectangular domain $\Omega := \left\{ \brac{x,y} \in
\brac{-0.5,1} \times \brac{-0.5, 1.5} \right\}$.  On $\pd \Omega$
Dirichlet boundary conditions for the velocity are specified according to
equation~\eqref{eqn:kovasnay-exact}. The pressure is prescribed in the
lower-left corner of the domain.  Equal-order polynomial elements are
used ($k = \Bar{k} = m = \Bar{m}$) with polynomial orders ranging from
one to five. The parameters $\chi = 1/2$, $\alpha = 6 k^2$ and $\beta =
10^{-4}$ are used.  We solve the stationary problem using a fixed point
iteration with stopping criterion
\begin{equation}
 \frac{\left| e_u^{i+1}-e_u^i \right|}{e_u^{i+1}+e_u^i} \leq {\rm TOL},
\end{equation}
where $e_u^{i}$ and $e_u^{i + 1}$ are the $L^2$ velocity error norms,
relative to the exact solution, of the consecutive iterates $i$ and
$i+1$, respectively, and ${\rm TOL}$ is a given tolerance which is set
to~$10^{-4}$.

For $Re = 40$, the observed convergence rates in the
$L^{2}$-norm for the velocity and pressure fields are presented
in Figure~\ref{fig:kovasnay_u-p}.  Convergence rates of order $k +
1$ for the velocity field and of order $k$ for the pressure field are
observed. It was verified that these convergence results also hold for
$\chi=0$ (advective scheme) and $\chi=1$ (conservative scheme).
\begin{figure}
  \center
  \begin{tabular}{c}
  \includegraphics[width=0.75\textwidth]{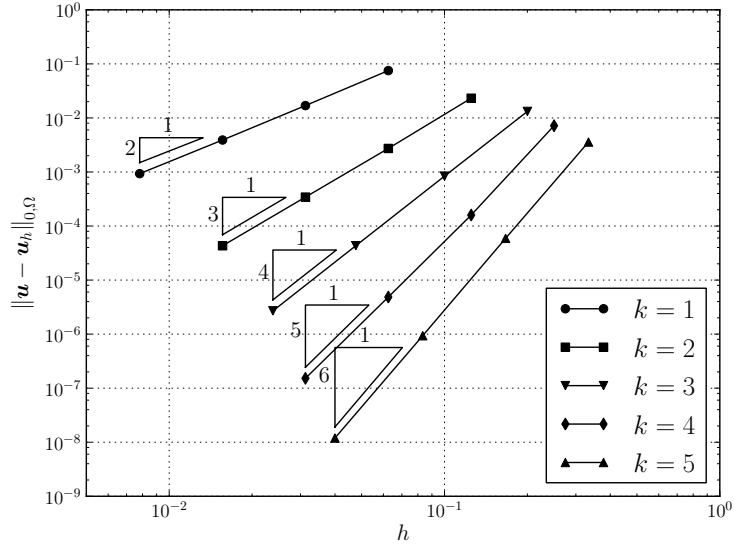}
  \\ (a) \\[1ex]
  \includegraphics[width=0.75\textwidth]{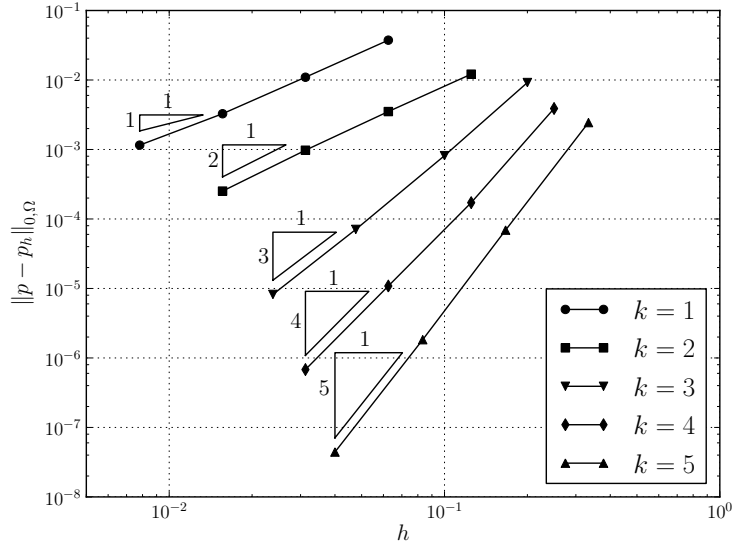}
  \\ (b)
 \end{tabular}
\caption{Kovasznay flow ($Re=40$): computed $L^{2}$ errors in (a)
 velocity and (b) pressure with $h$-refinement and for various
 polynomial orders $k$ ($\alpha=6k^2$, $\beta=10^{-4}$ and $\chi=1/2$).}
\label{fig:kovasnay_u-p}
\end{figure}

\subsection{Backward-facing step flow}

The next example concerns stationary two-dimensional flow over a
backward-facing step.  Figure~\ref{fig:backstep-geometry} presents the
set-up of the problem.  The step height $S$ is equal to half the height
of the main channel height $D$ and the velocity profile in the inflow
channel is parabolic with maximum velocity~$U_{\max}$.
\begin{figure}
  \center
  \includegraphics[width=0.6\textwidth]{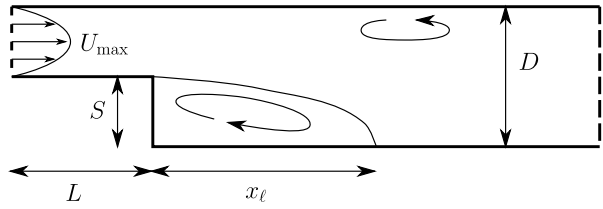}
\caption{Backward-facing step: general set-up.}
\label{fig:backstep-geometry}
\end{figure}
Behind the step a recirculation zone develops with the re-attachment
length $x_{\ell}$ depending on the Reynolds-number, as investigated
experimentally by \citet{Armaly:1983}. The Reynolds number is defined as
\begin{equation}
 Re = \frac{UD}{\nu},
 \label{eqn:Re-backstep}
\end{equation}
where $U$ is two-thirds of the maximum inflow velocity and $\nu$ is the
kinematic viscosity~\cite{Armaly:1983}.

The numerical test concerns the comparison of experimental
and computed values of the dimensionless reattachment length
$x_{\ell}/S$~\cite{Armaly:1983}.  We consider zero step length ($L=0$),
and a rectangular computational domain $\Omega := \left\{ \brac{x,y}
\in \brac{0,15} \times \brac{0, 1} \right\}$ which extends from the step
over a length of $30$ times the step height in the downstream direction.
The domain is partitioned using $301$ vertices in $x$-direction and 31
vertices in $y$-direction.  On the left boundary ($x=0$) the parabolic
velocity profile with $U_{\max}=1$ is imposed for $1/2 \leq y \leq 1$
using a Dirichlet boundary condition.  Along the outflow boundary ($x=15$)
a homogeneous Neumann boundary condition for the velocity is used.
On all other boundaries $\bar{\vect{u}} = \vect{0}$.  The pressure
degree of freedom for $\Bar{p}$ in the lower-left corner of the domain
is fixed. Using the definition in~\eqref{eqn:Re-backstep}, the maximum
inflow velocity and kinematic viscosity are adjusted to obtain a range of
Reynolds numbers between $100$ and~$800$. Equal order polynomials are used
($k = \Bar{k} = m = \Bar{m}$) with $k = 1$ and $k = 2$ considered. The
parameters $\chi=1/2$, $\alpha=6 k^2$ and $\beta = 10^{-4}$ are used.
We solve the stationary problem using a fixed point iteration with a
stopping criterion based on the $L^2$-norm of the velocity,
\begin{equation}
 \frac{\norm{\vect{u}_h}_{0, \Omega}^{i+1}
  - \norm{\vect{u}_h}_{0, \Omega}^{i}}
  {\norm{\vect{u}_h}_{0, \Omega}^{i+1}
  + \norm{\vect{u}_h}_{0, \Omega}^{i}} \leq {\rm TOL},
\end{equation}
where $i$ and $i + 1$ denote successive iterates and ${\rm TOL}$
is a tolerance, which is set to~$10^{-6}$.

The computed dimensionless reattachment lengths for various Reynolds
numbers are presented in Figure~\ref{fig:backstep-results}, and are
compared against measured data~\cite{Armaly:1983}.
\begin{figure}
  \center
  \includegraphics[width=0.75\textwidth]{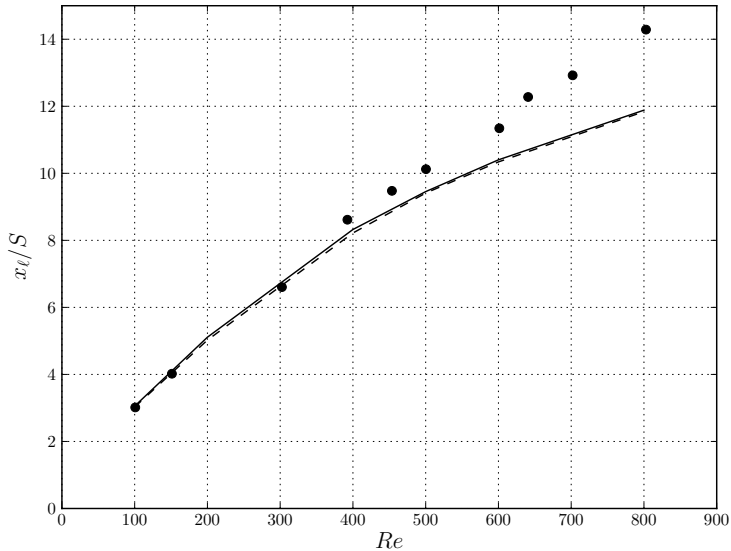}
\caption{Backward-facing step: comparison of measured ($\bullet$) and
 computed reattachment lengths for polynomial orders of $k=1$ (solid)
 and $k=2$ (dashed), experimental data from~\cite{Armaly:1983}.}
\label{fig:backstep-results}
\end{figure}
For $Re < 400$ the computed results are in good agreement with the
results obtained from the experiments. For $Re > 400$ the computed
results gradually deviate from the measurements, in a similar way as the
results computed by~\citet{Kim:1985}, and which can be attributed to the
emergence of three-dimensional flow structures~\citep{Armaly:1983}.
The computed streamlines for $Re = 800$ and polynomial orders of
one are shown in Figure~\ref{fig:backstep-streamlines}. The computed
streamlines involve a secondary recirculation bubble which resides
between dimensionless distances of $10.4$ and $20.1$ from the step,
which is in good agreement with observed values of $11.2$ and $19.6$,
respectively~\cite{Armaly:1983}.
\begin{figure}
  \center
  \includegraphics[width=0.95\textwidth]{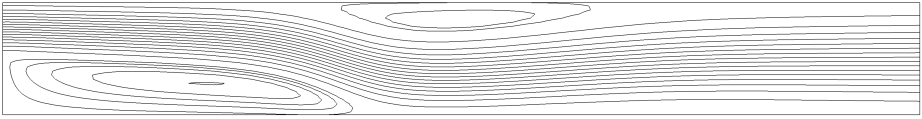}
\caption{Backward-facing step: computed streamlines for $Re=800$ and $k=1$,
 stream function intervals $0.2$ (main flow), $0.005$ (recirculation zone)
 and $0.002$ (secondary bubble), respectively,
 figure stretched by a factor of two in the cross-stream direction.}
\label{fig:backstep-streamlines}
\end{figure}

\subsection{Chaotic advection}

We now consider the energy stability properties of the method
for the incompressible Navier--Stokes equations with $\chi = 1/2$.
The incompressible Navier--Stokes equations are solved on the unit square
with zero viscosity and boundary conditions $\vect{u} \cdot \vect{n} = 0$
and $\vect{h} \cdot \vect{s} = 0$ on $\pd \Omega$, where $\vect{s} \cdot
\vect{n} = 0$. This corresponds to impermeable free-slip boundaries.
The pressure is prescribed to be zero at a point in the domain. The
initial condition $\vect{u}_{0} = \vect{0}$ is used. A time step $\delta
t = 0.2$ is adopted and the mesh has 32 cell vertices along each axis.
To create a chaotic velocity field, in the first simulation step a random
forcing term $\vect{f}$ is used.  Uniform random variables are generated
at vertices such that for each component of the forcing vector $f_{i}
\in [-1, 1]$.  This field is then interpolated using linear Lagrange
finite element basis functions.  For the first step, $\nu = 1 \times
10^{-5}$, after which it is set to zero.  This is done to start the
simulation, since with $\vect{u}_{0} = \vect{0}$, if $\nu = 0$ then
\eqref{eqn:fe_problem_discrete} cannot be solved as the advective
velocity at $t = 0$ is zero.  For the first $5$ time steps, a backward
Euler scheme is used ($\theta = 1$) to damp oscillations due to the
discontinuous nature of the forcing term. After the first $5$ steps,
$\theta = 1/2$ is used.

The relative change in the total kinetic energy between steps once $\theta
= 1/2$ is presented in Figure~\ref{fig:total_Ek} for the case of linear
basis functions ($k = 1$) for all fields and the case of quadratic basis
functions ($k = 2$) for all fields.  Consistent with the analysis, the
kinetic energy is observed to decrease monotonically. Not unexpectedly,
the relative dissipation is smaller for the $k=2$ case.
\begin{figure}
  \center
  \includegraphics[width=0.75\textwidth]{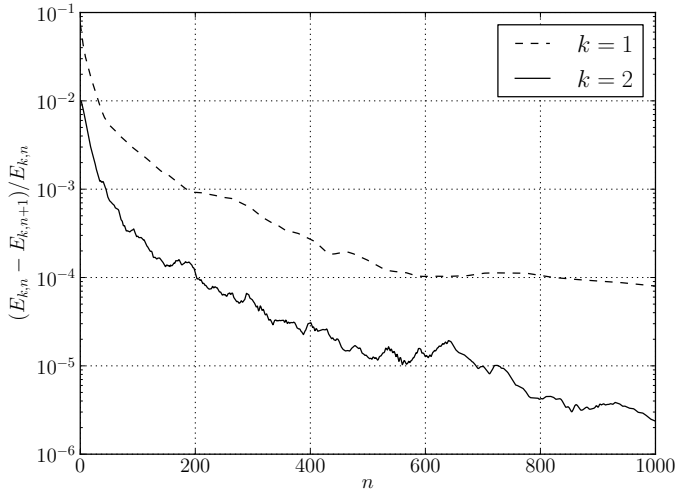}
\caption{Relative change in total kinetic energy
between time steps for the incompressible Navier--Stokes test
with $\chi = 1/2$, $\theta = 1/2$ and $\nu = 0$ for
linear ($k=1$) and quadratic ($k=2$) elements. In both cases, a mesh with
$32 \times 32$ vertices is used.}
\label{fig:total_Ek}
\end{figure}

\section{Conclusions}

A generalization of a hybrid method that inherits attractive properties
of continuous and discontinuous Galerkin methods has been presented and
analyzed for the incompressible Navier--Stokes equations.  The method
incorporates upwinding of the advective momentum flux naturally,
it is observed to be stable for equal-order velocity/pressure basis
functions and it has very good local mass conservation properties.
These properties, usually associated with discontinuous Galerkin methods,
can be achieved with the same number of global degrees of freedom as a
continuous Galerkin method on the same mesh, thereby obviating the common
criticism of discontinuous Galerkin methods that the number of degrees
of freedom is too large compared to continuous methods.  In contrast with
our earlier work, the new formulation presented here involves a pressure
field that is discontinuous across cell facets. This has implications for
local mass conservation, which in the presented formulation is guaranteed
in terms of the numerical flux.  It is shown that with appropriately
chosen (equal order) function spaces the method conserves momentum.
Moreover, the new formulation presented in this work uses a skew-symmetric
form of the momentum advection term.  It has been shown that this, in
combination with a suitable time integration scheme, guarantees that the
global kinetic energy will decay monotonically, even if the velocity
field is not point-wise divergence-free. The properties of the method
that have been demonstrated by analysis are supported by numerical
examples. Standard convergence rates for a range of polynomial orders
are observed in the Stokes and Navier--Stokes examples, and simulations
comparing the continuous and discontinuous pressure cases illustrate the
advantage of discontinuous pressure fields for local mass conservation.
The Navier--Stokes example concerning the flow over a backward facing
step shows that the method performs well in an advection dominated
case. The Navier--Stokes example concerning the evolution of a randomly
generated velocity field demonstrates the energy decaying property of
the skew-symmetric momentum advection term.  The complete computer code
for performing all presented numerical examples is made freely available
under an open source license as part of the supporting material.

\bibliographystyle{unsrtnat}
\bibliography{references}
\end{document}